\documentclass[a4paper,UKenglish]{article}

\newcommand{\ls}[2]{#1}

\ls{\usepackage{amsthm}
	\newtheorem{lemma}{Lemma}
	\usepackage{hyperref, authblk}
}{
	
	\usepackage{helvet}         
	\usepackage{courier}        
	\usepackage{makeidx}         
	\usepackage{graphicx}        
	\usepackage{multicol}        
	\usepackage[bottom]{footmisc}
	
	\makeindex             
	\usepackage{hyperref}
} 
\usepackage{todonotes}

\usepackage{tikz}
\usepackage{xcolor}
\usetikzlibrary{shapes.misc, graphs, positioning, shapes, shadows}
\usetikzlibrary{trees}
\usetikzlibrary{arrows}

\tikzstyle{bag} = [text width=4em, text centered]
\tikzstyle{end} = [circle, minimum width=3pt,fill, inner sep=0pt]

\tikzset{
  treenode/.style = {align=center, inner sep=0pt, text centered,
    font=\sffamily},
  arn_n/.style = {treenode, rounded rectangle, black, font=\sffamily\bfseries, draw=black,
    inner sep=0.15cm},
  arn_red/.style = {treenode, rounded rectangle, black, font=\sffamily\bfseries, draw=red,
    inner sep=0.15cm},
    arn_white/.style = {treenode, rounded rectangle, black, font=\sffamily\bfseries, draw=white,
    inner sep=0.15cm},
  arn_r/.style = {treenode, circle, red, draw=red, 
    text width=2em, very thick},
  arn_x/.style = {treenode, rectangle, draw=black,
    minimum width=0.5em, minimum height=0.5em}
}

\tikzset{
diagonal fill/.style 2 args={fill=#2, path picture={
\fill[#1, sharp corners] (path picture bounding box.south west) -|
                         (path picture bounding box.north east) -- cycle;}},
reversed diagonal fill/.style 2 args={fill=#2, path picture={
\fill[#1, sharp corners] (path picture bounding box.north west) |- 
                         (path picture bounding box.south east) -- cycle;}}
}

\newcommand{\figeins}{
\begin{tikzpicture}
  \tikzset{circle/.style={draw,circle,minimum width=0.01cm,opacity=0.4}} 
  
  \draw (0,0) ellipse (7cm and 2cm);  
  \draw (-1,-2) -- (-1,2);
  \draw (1,-2) -- (1,2);
  \draw (-3.5,-3) ellipse (3cm and 0.7cm); 
  \draw (3,-3) ellipse (3cm and 0.7cm);  
  
  \node at (-6,0) {\huge A};
  \node at (-1.3,-3) {\huge B};  
  \node at (-0.5,1.5) {\huge C}; 
  \node at (1.5,1.5) {\huge D}; 
  \node at (1,-3) {\huge F}; 
  \node at (6,2) {\huge S}; 
  
  \node[shape=circle,draw=black, fill = black, inner sep = 0.07cm] (a) at (-4,1) {};
  \node[shape=circle,draw=black, fill = black, inner sep = 0.07cm] (c) at (-4,0.5) {};
  \node[shape=circle,draw=black, fill = black, inner sep = 0.07cm] (e) at (-4,0) {};
  \node[shape=circle,draw=black, fill = black, inner sep = 0.07cm] (b) at (-2.5,1) {};
  \node[shape=circle,draw=black, fill = black, inner sep = 0.07cm] (d) at (-2.5,0.5) {};
  \node[shape=circle,draw=black, fill = black, inner sep = 0.07cm] (f) at (-2.5,0) {};
  
  \node[shape=circle,draw=black, fill = black, inner sep = 0.07cm] (g) at (-5,-1) {};
  \node[shape=circle,draw=black, fill = black, inner sep = 0.07cm] (h) at (-4,-1) {};
  \node[shape=circle,draw=black, fill = black, inner sep = 0.07cm] (i) at (-3,-1) {};
  \node[shape=circle,draw=black, fill = black, inner sep = 0.07cm] (j) at (-2,-1) {};
  \node[shape=circle,draw=black, fill = black, inner sep = 0.07cm] (k) at (-5,-3) {};
  \node[shape=circle,draw=black, fill = black, inner sep = 0.07cm] (l) at (-4,-3) {};
  \node[shape=circle,draw=black, fill = black, inner sep = 0.07cm] (m) at (-3,-3) {};
  \node[shape=circle,draw=black, fill = black, inner sep = 0.07cm] (n) at (-2,-3) {};
  
  \node[shape=circle,draw=black, fill = black, inner sep = 0.07cm] (o) at (-0.5,0) {};
  \node[shape=circle,draw=black, fill = black, inner sep = 0.07cm] (p) at (-0,0) {};
  \node[shape=circle,draw=black, fill = black, inner sep = 0.07cm] (q) at (0.5,0) {};
  
  \node[shape=circle,draw=black, fill = black, inner sep = 0.07cm] (r) at (5,0) {};
  \node[shape=circle,draw=black, fill = black, inner sep = 0.07cm] (s) at (5.5,0) {};
  \node[shape=circle,draw=black, fill = black, inner sep = 0.07cm] (t) at (2.5,0) {};
  \node[shape=circle,draw=black, fill = black, inner sep = 0.07cm] (u) at (3,0) {};
  \node[shape=circle,draw=black, fill = black, inner sep = 0.07cm] (v) at (3.5,0) {};
  \node[shape=circle,draw=black, fill = black, inner sep = 0.07cm] (w) at (4,0) {};
  \node[shape=circle,draw=black, fill = black, inner sep = 0.07cm] (x) at (4.5,0) {};
  
  \node[shape=circle,draw=black, fill = black, inner sep = 0.07cm] (y) at (4,-3) {};
  \node[shape=circle,draw=black, fill = black, inner sep = 0.07cm] (z) at (2.5,-3) {};
  \node[shape=circle,draw=black, fill = black, inner sep = 0.07cm] (y1) at (3,-3) {};
  \node[shape=circle,draw=black, fill = black, inner sep = 0.07cm] (x1) at (3.5,-3) {};

  \path [-] (a) edge node[] {} (b);
  \path [-] (c) edge node[] {} (d);
  \path [-] (e) edge node[] {} (f);
  \path [-] (g) edge node[] {} (k);
  \path [-] (h) edge node[] {} (l);
  \path [-] (i) edge node[] {} (m);
  \path [-] (j) edge node[] {} (n);

\end{tikzpicture}
}

\newcommand{\graphg}{
\begin{tikzpicture}[inner sep=0.1cm]
	\node[shape=circle,draw=black] (a) at (0,0) {a};
	\node[shape=circle,draw=black] (b) at (2,0) {b};
    \node[shape=circle,draw=black] (c) at (0,-1.5) {c};
    \node[shape=circle,draw=black] (d) at (2,-1.5) {d};
    \node[shape=circle,draw=black] (e) at (0,-3) {e};
    \node[shape=circle,draw=black] (f) at (2,-3) {f};

    \path [-] (c) edge node[] {} (a);
    \path [-] (b) edge node[] {} (a);
    \path [-] (c) edge node[] {} (e);
    \path [-] (c) edge node[] {} (d);
    \path [-] (d) edge node[] {} (b);
    \path [-] (d) edge node[] {} (f);
    
\end{tikzpicture}
}

\newcommand{\tdg}{
\begin{tikzpicture}[inner sep=0.2cm]
\node [arn_n] {$a,b,c$} 
	child{ node [arn_n] {$b,c,d$}          
			child{node[arn_n] {$c,e$} }
			child{node[arn_n] {$d,f$} }                                  
    }
; 
\end{tikzpicture}
}

\newcommand{\tdr}{
\begin{tikzpicture}
\node [arn_n] {$c$} 
	child{ node [arn_n] {$c,d$}          
			child{node[arn_n] {$c,\textcolor{red}{e}$} }
			child{node[arn_n] {$d,\textcolor{red}{f}$} }                                  
    }
; 
\end{tikzpicture}
}

\newcommand{\tdrb}{
\begin{tikzpicture}
\node [arn_n] {$c {|} 1$} 
	child{ node [arn_n] {$c,d|2$}          
			child{node[arn_n] {$c|2$} }
			child{node[arn_n] {$d|2$} }                                  
    }
; 
\end{tikzpicture}
}

\newcommand{\tdzweirb}{
\begin{tikzpicture}
\node [arn_n] {$c,b|2$} 
	child{ node [arn_n] {$c,d|4$}          
			child{node[arn_n] {$c,d|3$} 
				child{node[arn_n] {$c,d|2$} 
					child{node[arn_red] {$c|2$} }
					child{node[arn_n] {$c,d|4$} 
						child{node[arn_n] {$c,f|1$} }	
					}
				}
				child{node[arn_red] {$d|5$} 
					child[missing] {node {}}
					child{node[arn_red] {$d|3$} }
				}			
			}
			child{node[arn_n] {$d,a|3$} }                                  
    }
; 
\end{tikzpicture}
}

\newcommand{\tdzweic}{
\begin{tikzpicture}
\node [arn_n] {$c,b|2$} 
	child{ node [arn_n] {$c,d|4$}          
			child{node[arn_n] {$c,d|3$} 
				child{node[arn_n] {$c,d|2$} 
					child{node[arn_n] {$c,d|4$} 
						child{node[arn_n] {$c,f|2$} }	
					}
				}			
			}
			child{node[arn_n] {$d,a|3$} }                                  
    }
; 
\end{tikzpicture}
}

\newcommand{\tdzweicc}{
\begin{tikzpicture}
\node [arn_n] {$c,b|2$} 
	child{ node [arn_n] {$c,d|4$}          
			child{node[arn_n] {$c,d|3,2,4$} 
				child{node[arn_n] {$c,f|2$} }			
			}
			child{node[arn_n] {$d,a|3$} }                                  
    }
; 
\end{tikzpicture}
}

\newcommand{\ceins}{
\begin{tikzpicture}
\node [arn_n, label={[label distance=0.001cm]80:$p(t)$}] {$c,d|2$} 
	child{ node [arn_n] {$d|1$} }         
	child{node[arn_white] {$\textcolor{white}{c|3}$} }
	node[shape=rounded rectangle,draw=green,inner sep=0.2cm] (a) at (1.4,-1.5) {\textcolor{green}{Z}}
	node[shape=rounded rectangle, draw = black,inner sep=0.2cm, fill = white, label={[label distance=0.001cm]80:$t$}] (a) at (0.78,-1.5) {$c|3$}
; 
\end{tikzpicture}
}

\newcommand{\czwei}{
\begin{tikzpicture}
\node [arn_n, label={[label distance=0.001cm]80:$p(t)$}] {$c,d|2$} 
	child{ node [arn_n] {$d|1$} } 
	child{ node [arn_n, , label={[label distance=0.001cm]80:$t$}] {$c|3$} }         
	child[draw=red]{node[arn_white] {$\textcolor{red}{c|1}$} }
	node[shape=rounded rectangle,draw=green,inner sep=0.2cm] (a) at (2.15,-1.5) {\textcolor{green}{Z}}
	node[shape=rounded rectangle, draw = red,inner sep=0.2cm, fill = white, label={[label distance=0.001cm]80:$t'$}] (a) at (1.5,-1.5) {$\textcolor{red}{c|1}$}
; 
\end{tikzpicture}
}

\newcommand{\ceinss}{
	\begin{tikzpicture}
	\node [arn_n] {$c$} 
	child{ node [arn_n, label={[label distance=0.001cm]80:$p(t)$}] {$c,d|2$} 
	child{ node [arn_n] {$d|1$} }         
	child{node[arn_white] {$\textcolor{white}{c|3}$} } }
	node[shape=rounded rectangle,draw=green,inner sep=0.2cm] (a) at (1.4,-3) {\textcolor{green}{Z}}
	node[shape=rounded rectangle, draw = black,inner sep=0.2cm, fill = white, label={[label distance=0.001cm]80:$t$}] (a) at (0.78,-3) {$c|3$}
	; 
	\end{tikzpicture}
}

\newcommand{\czweis}{
	\begin{tikzpicture}
	\node [arn_n] {$c$}
	child { node [arn_n, label={[label distance=0.001cm]80:$p(t)$}] {$c,d|2$} 
	child{ node [arn_n] {$d|1$} }
	child{ node [arn_n, , label={[label distance=0.001cm]80:$t$}] {$c|3$} }        
	child[draw=red]{node[arn_white] {$\textcolor{red}{c|1}$} } }
	node[shape=rounded rectangle,draw=green,inner sep=0.2cm] (a) at (2.15,-3) {\textcolor{green}{Z}}
	node[shape=rounded rectangle, draw = red,inner sep=0.2cm, fill = white, label={[label distance=0.001cm]80:$t'$}] (a) at (1.5,-3) {$\textcolor{red}{c|1}$}
	; 
	\end{tikzpicture}
}

\newcommand{\cceins}{
\begin{tikzpicture}[inner sep=0.2cm]
	\node[shape=circle,draw=white] (nix) at (-1.5,0) {};
	\node[shape=circle,draw=green,inner sep = 0.1cm] (a2) at (0,0.8) {$\textcolor{green}{Z_1}$};
	\node[shape=circle,draw=black, fill = white] (a) at (0,0) {$c|1$};
	
	\node[shape=circle,draw=green,inner sep = 0.1cm] (b2) at (2.5,0.8) {$\textcolor{green}{Z_2}$};
	\node[shape=circle,draw=black, fill=white] (b) at (2.5,0) {$c|4$};
	
	\node[shape=circle,draw=green,inner sep = 0.1cm] (c2) at (5,0.8) {$\textcolor{green}{Z_3}$};
    \node[shape=circle,draw=black, fill=white] (c) at (5,0) {$c|3$};
    
    \node[shape=circle,draw=green,inner sep = 0.1cm] (d2) at (7.5,0.8) {$\textcolor{green}{Z_4}$};
    \node[shape=circle,draw=black, fill=white] (d) at (7.5,0) {$c|5$};
    \node[shape=circle,draw=white] (e) at (9,0) {};
    
    \node[shape=rectangle,draw=white] (f) at (-0.75,-2) {$u_1$};
    \node[shape=rectangle,draw=white] (g) at (0.75,-2) {$u_2$};
    \node[shape=rectangle,draw=white] (h) at (1.75,-2) {$\textcolor{blue}{u_3}$};
    \node[shape=rectangle,draw=white] (i) at (3.25,-2) {$\textcolor{red}{u_4}$};
    \node[shape=rectangle,draw=white] (j) at (4.25,-2) {$\textcolor{blue}{u_5}$};
    \node[shape=rectangle,draw=white] (k) at (5,-2) {$\textcolor{blue}{u_6}$};
    \node[shape=rectangle,draw=white] (l) at (5.75,-2) {$\textcolor{red}{u_7}$};
    \node[shape=rectangle,draw=white] (m) at (6.75,-2) {$\textcolor{blue}{u_8}$};
    \node[shape=rectangle,draw=white] (n) at (8.25,-2) {$\textcolor{red}{u_9}$};

    \path [-] (a) edge node[] {} (b);
    \path [-] (b) edge node[] {} (c);
    \path [-] (c) edge node[] {} (d);
    \path [dashed, -] (d) edge node[] {} (e);
    \path [dashed, -] (a) edge node[] {} (nix);
    
    \path [-] (a) edge node[] {} (g);
    \path [-] (a) edge node[] {} (f);
    
    \path [blue, -] (b) edge node[] {} (h);
    \path [red,-] (b) edge node[] {} (i);
    
    \path [blue,-] (c) edge node[] {} (j);
    \path [blue,-] (c) edge node[] {} (k);
    \path [red,-] (c) edge node[] {} (l);    
    
    \path [blue,-] (d) edge node[] {} (m);
    \path [red,-] (d) edge node[] {} (n);

\end{tikzpicture}
}

\newcommand{\cczwei}{
\begin{tikzpicture}[inner sep=0.2cm]
	\node[shape=circle,draw=white] (nix) at (-1.5,0) {};
	\node[shape=circle,draw=green, inner sep = 0.1cm] (a2) at (0,0.8) {$\textcolor{green}{Z_1}$};
	\node[shape=circle,draw=black, fill=white] (a) at (0,0) {$c|1$};
	
	\node[shape=circle,draw=green, inner sep = 0.1cm] (a2) at (2,0.8) {$\textcolor{green}{Z_2}$};
	\node[shape=circle,draw=black, fill=white] (b) at (2,0) {$c|1$};
	
	\node[shape=circle,draw=green, inner sep = 0.1cm] (a2) at (4,0.8) {$\textcolor{green}{Z_3}$};
    \node[shape=circle,draw=black, fill=white] (c) at (4,0) {$c|1$};
    
    \node[shape=circle,draw=green, inner sep = 0.1cm] (a2) at (6,0.8) {$\textcolor{green}{Z_4}$};
    \node[shape=circle,draw=black, fill=white] (d) at (6,0) {$c|1$};
    
    \node[shape=circle,draw=green, inner sep = 0.1cm] (a2) at (8,0.8) {$\textcolor{green}{Z_4}$};
    \node[shape=circle,draw=black, fill=white] (x) at (8,0) {$c|4$};
    
    \node[shape=circle,draw=green, inner sep = 0.1cm] (a2) at (10,0.8) {$\textcolor{green}{Z_4}$};
    \node[shape=circle,draw=black, fill=white] (y) at (10,0) {$c|3$};
    
    \node[shape=circle,draw=green, inner sep = 0.1cm] (a2) at (12,0.8) {$\textcolor{green}{Z_4}$};
    \node[shape=circle,draw=black, fill=white] (z) at (12,0) {$c|5$};
    \node[shape=circle,draw=white] (e) at (13.5,0) {};
    
    \node[shape=rectangle,draw=white] (f) at (-0.5,-2) {$u_1$};
    \node[shape=rectangle,draw=white] (g) at (0.5,-2) {$u_2$};
    \node[shape=rectangle,draw=white] (h) at (2,-2) {$\textcolor{blue}{u_3}$};
    \node[shape=rectangle,draw=white] (j) at (3.5,-2) {$\textcolor{blue}{u_5}$};
    \node[shape=rectangle,draw=white] (k) at (4.5,-2) {$\textcolor{blue}{u_6}$};
    \node[shape=rectangle,draw=white] (m) at (6,-2) {$\textcolor{blue}{u_8}$};
    \node[shape=rectangle,draw=white] (n) at (8,-2) {$\textcolor{red}{u_4}$};
    \node[shape=rectangle,draw=white] (o) at (10,-2) {$\textcolor{red}{u_7}$};
    \node[shape=rectangle,draw=white] (p) at (12,-2) {$\textcolor{red}{u_9}$};

    \path [-] (a) edge node[] {} (b);
    \path [-] (b) edge node[] {} (c);
    \path [-] (c) edge node[] {} (d);
    \path [-] (d) edge node[] {} (x);
    \path [-] (x) edge node[] {} (y);
    \path [-] (y) edge node[] {} (z);
    \path [dashed, -] (z) edge node[] {} (e);
    \path [dashed, -] (a) edge node[] {} (nix);
    
    \path [-] (a) edge node[] {} (g);
    \path [-] (a) edge node[] {} (f);
    
    \path [blue, -] (b) edge node[] {} (h);
    
    \path [blue,-] (c) edge node[] {} (j);
    \path [blue,-] (c) edge node[] {} (k);
    
    \path [blue,-] (d) edge node[] {} (m);
    
    \path [red,-] (x) edge node[] {} (n);
   
    \path [red,-] (y) edge node[] {} (o);
    
    \path [red,-] (z) edge node[] {} (p);

\end{tikzpicture}
}

\newcommand{\gab}{
\begin{tikzpicture}[inner sep=0.2cm]
	\node[shape=circle,draw=white] (a) at (1,0) {5};
	\node[shape=circle,draw=black] (b) at (2,0) {6};
	\node[shape=circle,draw=black] (c) at (4,0) {8};
	\node[shape=circle,draw=black] (d) at (6,0) {7};
	\node[shape=circle,draw=black] (e) at (8,0) {9};
	
	\node[shape=circle,draw=white] (f) at (1,-2) {2};
	\node[shape=circle,draw=black] (g) at (2,-2) {3};
	\node[shape=circle,draw=black] (h) at (4,-2) {5};
	\node[shape=circle,draw=black] (i) at (6,-2) {4};
	\node[shape=circle,draw=black] (j) at (8,-2) {6};
	
	\node[shape=circle,draw=white] (k) at (1,-4) {4};
	\node[shape=circle,draw=black] (l) at (2,-4) {5};
	\node[shape=circle,draw=black] (m) at (4,-4) {7};
	\node[shape=circle,draw=black] (n) at (6,-4) {6};
	\node[shape=circle,draw=black] (o) at (8,-4) {8};
	
	\node[shape=circle,draw=white] (p) at (2,-5) {1};
	\node[shape=circle,draw=white] (q) at (4,-5) {3};
	\node[shape=circle,draw=white] (r) at (6,-5) {2};
	\node[shape=circle,draw=white] (s) at (8,-5) {4};

    \path [->] (b) edge node[] {} (c);
    \path [->] (c) edge node[] {} (d);
    \path [->, red] (d) edge node[] {} (e);

    \path [->] (g) edge node[] {} (h);
    \path [->] (h) edge node[] {} (i);
    \path [->] (i) edge node[] {} (j);
    
    \path [->,red] (l) edge node[] {} (m);
    \path [->] (m) edge node[] {} (n);
    \path [->] (n) edge node[] {} (o);
    
    \path [->] (g) edge node[] {} (b);
    \path [->] (h) edge node[] {} (c);
    \path [->,red] (i) edge node[] {} (d);
    \path [->] (j) edge node[] {} (e);
    
    \path [->] (l) edge node[] {} (g);
    \path [->] (m) edge node[] {} (h);
    \path [->] (n) edge node[] {} (i);
    \path [->] (o) edge node[] {} (j);
    
    \path [->] (g) edge node[] {} (c);
    \path [->] (h) edge node[] {} (d);
    \path [->] (i) edge node[] {} (e);
    
    \path [->] (l) edge node[] {} (h);
    \path [->,red] (m) edge node[] {} (i);
    \path [->] (n) edge node[] {} (j);

\end{tikzpicture}
}

\newcommand{\topologie}{
\begin{tikzpicture}
  \tikzset{circle/.style={draw,circle,minimum width=0.01cm,opacity=0.4}} 
    
  \node[shape=circle,draw=black, fill = blue, inner sep = 0.09cm, label=$\textcolor{blue}{t_1}$] (a) at (0,0) {};
  \node[shape=circle,draw=black, fill = blue, inner sep = 0.09cm] (b) at (0,-0.7) {};
  \node[shape=circle,draw=black, fill = blue, inner sep = 0.09cm] (c) at (1,-0.7) {};
  \node[shape=circle,draw=black, fill = blue, inner sep = 0.09cm] (d) at (0,-1.4) {};
  \node[shape=circle,draw=black, fill = black, inner sep = 0.09cm] (d1) at (-1,0.6) {};
  \node[shape=circle,draw=black, fill = black, inner sep = 0.09cm] (e) at (-1,1.4) {};
  \node[shape=circle,draw=black, fill = black, inner sep = 0.09cm] (f) at (-1.8,0.6) {};

  \node[shape=circle,draw=black, fill = green, inner sep = 0.09cm] (g) at (0.6,0.4) {};
  \node[shape=circle,draw=black, fill = green, inner sep = 0.09cm] (h) at (1.2,0.8) {};
  \node[shape=circle,draw=black, fill = red, inner sep = 0.09cm, label=$\textcolor{red}{t_2}$] (i) at (1.8,1.2) {};
  \node[shape=circle,draw=black, fill = red, inner sep = 0.09cm] (j) at (2.3,1.7) {};
  \node[shape=circle,draw=black, fill = red, inner sep = 0.09cm] (k) at (2.3,0.8) {};
  \node[shape=circle,draw=black, fill = black, inner sep = 0.09cm] (x) at (2.3,0) {};
  
  \node[shape=circle,draw=black, fill = blue, inner sep = 0.09cm] (l) at (-0.6,-1.9) {};
  \node[shape=circle,draw=black, fill = blue, inner sep = 0.09cm] (l2) at (-0.6,-2.4) {};
  \node[shape=circle,draw=black, fill = blue, inner sep = 0.09cm] (m) at (-1.2,-1.9) {};
  \node[shape=circle,draw=black, fill = black, inner sep = 0.09cm] (n) at (-1.6,-1.5) {};  
  \node[shape=circle,draw=black, fill = black, inner sep = 0.09cm] (o) at (-1.6,-2.3) {};
  
  \node[shape=circle,draw=black, fill = black, inner sep = 0.09cm] (p) at (0.8,-1.9) {};
  \node[shape=circle,draw=black, fill = black, inner sep = 0.09cm] (q) at (0.8,-2.4) {};  
  \node[shape=circle,draw=black, fill = black, inner sep = 0.09cm] (r) at (1.6,-1.9) {};
  \node[shape=circle,draw=black, fill = black, inner sep = 0.09cm] (s) at (1.6,-2.4) {};  
  \node[shape=circle,draw=black, fill = black, inner sep = 0.09cm] (t) at (2.2,-1.6) {};

  \path [-, blue] (a) edge node[] {} (b);
  \path [-, blue] (b) edge node[] {} (c);
  \path [-, blue] (b) edge node[] {} (d);
  \path [-, blue] (d) edge node[] {} (l);
  \path [-, blue] (l) edge node[] {} (m);  
  \path [-, blue] (l) edge node[] {} (l2);
  
  \path [-, red] (i) edge node[] {} (k);
  \path [-, red] (i) edge node[] {} (j);
  \path [-] (k) edge node[] {} (x);
  
  \path [-, green] (a) edge node[] {} (g);
  \path [-, green] (g) edge node[] {} (h);
  \path [-, green] (i) edge node[] {} (h);
  
  \path [-] (a) edge node[] {} (d1);
  \path [-] (d1) edge node[] {} (e);
  \path [-] (d1) edge node[] {} (f);
  
  \path [-] (m) edge node[] {} (n);
  \path [-] (m) edge node[] {} (o);
  
  \path [-] (d) edge node[] {} (p);
  \path [-] (p) edge node[] {} (q);
  \path [-] (p) edge node[] {} (r);
  \path [-] (r) edge node[] {} (s);
  \path [-] (r) edge node[] {} (t);
  
\end{tikzpicture}
}

\newcommand{\topologiezwei}{
\begin{tikzpicture}
  \tikzset{circle/.style={draw,circle,minimum width=0.01cm,opacity=0.4}} 
    
  \node[diagonal fill={red}{blue}, shape=circle,draw=black, inner sep = 0.09cm,] (a) at (0,0) {};
  \node[diagonal fill={red}{blue}, shape=circle,draw=black, inner sep = 0.09cm] (b) at (0,-0.7) {};
  \node[shape=circle,draw=black, fill = blue, inner sep = 0.09cm] (c) at (1,-0.7) {};
  \node[shape=circle,draw=black, fill = blue, inner sep = 0.09cm] (d) at (0,-1.4) {};
  \node[shape=circle,draw=black, fill = red, inner sep = 0.09cm] (d1) at (-1,0.6) {};
  \node[shape=circle,draw=black, fill = red, inner sep = 0.09cm] (e) at (-1,1.4) {};
  \node[shape=circle,draw=black, fill = red, inner sep = 0.09cm] (f) at (-1.8,0.6) {};

  \node[shape=circle,draw=black, fill = black, inner sep = 0.09cm] (f2) at (-1.8,-0.2) {};

  \node[shape=circle,draw=black, fill = black, inner sep = 0.09cm] (g) at (0.6,0.4) {};
  \node[shape=circle,draw=black, fill = black, inner sep = 0.09cm] (h) at (1.2,0.8) {};
  \node[shape=circle,draw=black, fill = black, inner sep = 0.09cm] (i) at (1.8,1.2) {};
  \node[shape=circle,draw=black, fill = black, inner sep = 0.09cm] (j) at (2.3,1.7) {};
  \node[shape=circle,draw=black, fill = black, inner sep = 0.09cm] (k) at (2.3,0.8) {};
  \node[shape=circle,draw=black, fill = black, inner sep = 0.09cm] (x) at (2.3,0) {};
  
  \node[shape=circle,draw=black, fill = blue, inner sep = 0.09cm] (l) at (-0.6,-1.9) {};
  \node[shape=circle,draw=black, fill = blue, inner sep = 0.09cm] (l2) at (-0.6,-2.4) {};
  \node[shape=circle,draw=black, fill = blue, inner sep = 0.09cm] (m) at (-1.2,-1.9) {};
  \node[shape=circle,draw=black, fill = black, inner sep = 0.09cm] (n) at (-1.6,-1.5) {};  
  \node[shape=circle,draw=black, fill = black, inner sep = 0.09cm] (o) at (-1.6,-2.3) {};
  
  \node[shape=circle,draw=black, fill = black, inner sep = 0.09cm] (p) at (0.8,-1.9) {};
  \node[shape=circle,draw=black, fill = black, inner sep = 0.09cm] (q) at (0.8,-2.4) {};  
  \node[shape=circle,draw=black, fill = black, inner sep = 0.09cm] (r) at (1.6,-1.9) {};
  \node[shape=circle,draw=black, fill = black, inner sep = 0.09cm] (s) at (1.6,-2.4) {};  
  \node[shape=circle,draw=black, fill = black, inner sep = 0.09cm] (t) at (2.2,-1.6) {};
  
  \path [-] (f) edge node[] {} (f2);

  \path [-, red] (a) edge node[] {} (b);
  \path [-, dashed, blue] (a) edge node[] {} (b);
  \path [-, blue] (b) edge node[] {} (c);
  \path [-, blue] (b) edge node[] {} (d);
  \path [-, blue] (d) edge node[] {} (l);
  \path [-, blue] (l) edge node[] {} (m);  
  \path [-, blue] (l) edge node[] {} (l2);
  
  \path [-] (a) edge node[] {} (g);
  \path [-] (g) edge node[] {} (h);
  \path [-] (i) edge node[] {} (h);
  \path [-] (i) edge node[] {} (k);
  \path [-] (i) edge node[] {} (j);
  \path [-] (k) edge node[] {} (x);
  
  \path [-, red] (a) edge node[] {} (d1);
  \path [-, red] (d1) edge node[] {} (e);
  \path [-, red] (d1) edge node[] {} (f);
  
  \path [-] (m) edge node[] {} (n);
  \path [-] (m) edge node[] {} (o);
  
  \path [-] (d) edge node[] {} (p);
  \path [-] (p) edge node[] {} (q);
  \path [-] (p) edge node[] {} (r);
  \path [-] (r) edge node[] {} (s);
  \path [-] (r) edge node[] {} (t);
  
\end{tikzpicture}
}

\ls{\bibliographystyle{plainurl}}{
	\bibliographystyle{spmpsci}
}

\newcommand{\figscale}{\ls{1}{0.70}}
\newcommand{\figspace}{\vspace*{-0.2em}}

\newcommand{\Oh}{\mathcal{O}}
\newcommand{\tw}{{tw}}
\newcommand{\given}{\hat}
\newcommand{\td}{\mathcal{T}}
\newcommand{\TD}{\mathcal{TD}}
\renewcommand{\L}{\mathcal{L}}
\ls{\newcommand{\E}{\mathcal{E}}}
{\renewcommand{\E}{\mathcal{E}}}
\newcommand{\ignore}[1]{}
\newcommand{\poly}{\rm poly}

\newtheorem{construction}{\bf Construction}

\begin{document}

\ls{
	\title{Optimal Tree Decompositions Revisited: A Simpler Linear-Time FPT Algorithm}
	
	\author[1]{Ernst Althaus}
	\author[1]{Sarah Ziegler}

	\affil[1]{Johannes Gutenberg-Universit\"at Mainz}	
                                                  

}{
\author{Ernst Althaus and Sarah Ziegler}
\authorrunning{E. Althaus and S. Ziegler}
\title*{Optimal Tree Decompositions Revisited: A Simpler Linear-Time FPT Algorithm}
\titlerunning{Optimal Tree Decompositions Revisited}

}

	\maketitle
	
	\begin{abstract}{
		In 1996, Bodlaender showed the celebrated result that an optimal tree decomposition of a graph of bounded treewidth can be found in linear time. The algorithm is based on an algorithm of Bodlaender and Kloks that computes an optimal tree decomposition given a non-optimal tree decomposition of bounded width. Both algorithms, in particular the second, are hardly accessible. \ls{In our review, we present them in a much simpler way than the original presentations.}{We present the second algorithm in a much simpler way in this paper and refer to an extended version for the first.} In our description of the second algorithm, we start by explaining how all tree decompositions of subtrees defined by the nodes of the given tree decomposition can be enumerated. We group tree decompositions into equivalence classes depending on the current node of the given tree decomposition, such that it suffices to enumerate one tree decomposition per equivalence class and, for each node of the given tree decomposition, there are only a constant number of classes which can be represented in constant space. \ls{Our description of the first algorithm further simplifies Perkovic and Reed's simplification.}{}}
	\end{abstract}

	\section{Introduction}
	
	
	
	
	
	
	Tree decompositions and treewidth are important concepts in parameterized complexity and are therefore introduced in many textbooks on graph-theory, graph algorithms, or parameterized complexity, e.g.,~\ls{\cite{DBLP:series/txtcs/FlumG06, niedermeier2006invitation, DBLP:series/txcs/DowneyF13,DBLP:books/daglib/0030488, Kloks13advancesin,DBLP:books/sp/CyganFKLMPPS15}}{\cite{DBLP:series/txtcs/FlumG06, niedermeier2006invitation, DBLP:series/txcs/DowneyF13,DBLP:books/sp/CyganFKLMPPS15}}. They are even introduced in the basic algorithms book of Kleinberg and Tardos \cite{DBLP:books/daglib/0015106}. Many NP-hard problems are fixed-parameter tractable in the treewidth -- i.e.,~for a graph $G=(V,E)$ of treewidth $\tw$, they can be solved in time $\Oh(f(\tw) \cdot \poly(|V|,|E|))$ for a computable function $f$ and a polynomial $\poly$. A necessary condition for these algorithms is that a tree decomposition with a similar width as the treewidth of the graph can be computed. In most textbooks, a (rather complicated) algorithm that computes a tree decomposition of width at most $4\tw$ is shown.
	
	Bodlaender \cite{DBLP:journals/siamcomp/Bodlaender96} proved that an optimal tree decomposition of a graph $G$ with fixed treewidth can be computed in linear time. Since its publication, this paper has been cited more than 1,500 times. The algorithm, however, does not appear in textbooks, as its presentation is way too complicated. We aim to give a simpler description in order to make the algorithm accessible to a wider audience.
	
	The algorithm is based on an algorithm by Bodlaender and Kloks \cite{DBLP:journals/jal/BodlaenderK96} that computes an optimal tree decomposition in linear time, if a (non-optimal) tree decomposition of fixed width is given.  
	\ls{One can compute the tree width exactly by applying this algorithm to the tree decomposition of width at most $4\tw$ that was constructed by the algorithm mentioned above. Since the algorithm to compute the non-optimal tree decomposition is not linear-time, the overall running time is not linear.}{}
	Bodlaender \cite{DBLP:journals/siamcomp/Bodlaender96} shows how, given a graph $G$, one can find a graph $G'$ of at most the same treewidth that is a constant factor smaller and it is easy to construct a tree decomposition for $G$ of width $2\tw$ from an optimal tree decomposition of $G'$. The algorithm has linear running time. Together with the algorithm of Bodlaender and Kloks, this gives an exact linear time algorithm:
	\begin{itemize}
		\item Compute $G'$.
		\item Compute an optimal tree decomposition of $G'$ recursively.
		\item Construct a tree decomposition of $G$ of width at most $2\tw$ from the tree decomposition of $G'$.
		\item Use the algorithm of Bodlaender and Kloks to find the optimal tree decomposition.
	\end{itemize}

	Several attempts were made to simplify the construction of an appropriate graph $G'$ (i.e., with the properties mentioned above) from $G$ (see e.g.,~\cite{DBLP:journals/ijfcs/PerkovicR00, DBLP:series/txcs/DowneyF13}). \ls{Hence, we mainly worked on the algorithm of Bodlaender and Kloks and devoted only a small section of this paper on the overall algorithm, since it basically follows the simplification of Perkovic and Reed \cite{DBLP:journals/ijfcs/PerkovicR00}.}{Hence, we only show the main ideas of our work on the algorithm of Bodlaender and Klocks in this extended abstract and refer to \cite{arxiv} for more details the overall algorithm.}
	
	\ls{To simplify the algorithm of Bodlaender and Kloks, we start with an easy algorithm that enumerates all tree decompositions using the given one of bounded width. More precisely, we assume that we are given a nice tree decomposition. We show how to construct all tree decompositions of the subgraphs induced by the vertices in the bags below a node of the given tree decomposition from the ones of the child nodes. We then group tree decompositions into equivalence classes depending on the current node of the given tree decomposition so that it suffices to enumerate one tree decomposition per class. This is done in three steps, coarsening the equivalence classes until only a constant number of classes per node of the given tree decomposition remain. Furthermore, we show that these classes can be represented in constant space and that these representations can be computed directly from the corresponding representations of the child nodes.}{}
	
\ls{The algorithm of Bodlaender has linear time if the treewidth of the given graph is bounded. More precisely, the running time is in $2^{\Oh(\tw^3)}|V(G)|$. Since the invention of this algorithm there was a lot of research (see \cite{BodlaenderDDFLP16} for an overview) on exact and approximate algorithms to improve the dependence on the treewidth as this is the bottleneck in the running time of many FPT-algorithms based on tree decompositions. Notice that if the FPT-algorithm has a running time of $f(\tw) \cdot \poly(n)$, the running time increases with a worse approximation ratio as $\tw$ has to be replaced by the guaranteed width of the tree decomposition found. Nevertheless, if we have an approximation algorithm with a significantly better running time dependence on the treewidth, this typically gives a better total running time of the overall algorithm if $f \in o(2^{\tw})$. Despite the intensive research, only recently Bodlaender et al.\cite{BodlaenderDDFLP16} gave an algorithm which finds a tree decomposition of width at most $5\tw+4$ with a running time that is linear in the size of the graph and single exponential in the treewidth.

}{}
	
	\section{Definitions and Basic Properties}
	\label{basic}
	
	\ls{
	In this Section, we give some definitions and sketch the proofs of some basic properties that we will use later in the paper. All definitions are standard and the properties are well known (see, e.g.,~\cite{DBLP:books/sp/CyganFKLMPPS15}).
	
	For a finite set $X$ we denote its power set by $2^X$ and for a function $f:X \rightarrow 2^Y$ and $X' \subseteq X$ let $f(X')=\bigcup_{x \in X'} f(x)$ (and similarly for multivariate functions).
	
	For an undirected graph $G=(V,E)$, let $V(G)=V$ be its vertices and $E(G)=E$ be its edges. Given an edge $uv \in E(G)$, shrinking the edge means replacing $G$ by the graph $G^{uv}$ with $V(G^{uv})=V(G) \setminus \{ u, v \} \cup \{ [uv] \}$, where $[uv] \notin V(G)$ is a new node and $E(G^{uv})$ is the set $E(G) \setminus \{ uv \}$ but endpoints $u$ or $v$ of the edges are replaced by $[uv]$. A minor of a graph $G$ is a graph that can be constructed from $G$ by shrinking or deleting some of the edges. For $V' \subseteq V$, the induced subgraph of $V$, denoted as $G[V']$, is the graph with vertex set $V'$ and edge set $E'=\{uv \in E \mid u,v \in V' \}$.
	A tree is a connected undirected graph without simple cycles.

	A tree decomposition $(T, (X_t)_{t \in T})$ for a graph $G$ is a tuple of a tree $T$ over some set of vertices $V(T)$ and subsets of vertices $X_t \subseteq V(G)$ of $G$, one for each vertex in $T$, such that the following three properties hold:
	\begin{description}
		\item[(Node coverage)] For each $v \in V(G)$ there is at least one $t \in V(T)$ such that $v \in X_t$.
		\item[(Edge coverage)] For each $uv \in E(G)$ there is at least one $t \in V(T)$ such that $u$ and $v$ are in $X_t$.
		\item[(Coherence)] If $v \in X_{t_1}$ and $v \in X_{t_2}$ for $v \in V(G)$ and $t_1, t_2 \in V(T)$ then $v \in X_{t_3}$ for all vertices $t_3$ on the unique path in $T$ from $t_1$ to $t_2$.
	\end{description}
	
	The sets $X_t$ are called the bags.
	The width of a tree decomposition $(T, (X_t)_{v \in V(T)})$ is the maximal cardinality of one of the sets $X_t$ minus one, i.e., $\max_{t \in V(T)} |X_t|-1$. The treewidth of a graph $G$ is the minimal width of a tree decomposition of $G$.
	
	To make a clearer distinction between the vertices $V(G)$ of the graph and those of the tree of the tree decomposition, we call the elements of $V(T)$ nodes and the elements of $V(G)$ vertices in the following. Notice that for all $V'=\{t \in V(T) \mid v \in X_t \}$ for some $v \in V(G)$ the coherence implies that $T[V']$ is a tree (i.e., connected). On the other hand, if $T[V']$ is a tree for all of those subsets, this implies coherence. Furthermore, if each node of a graph $G$ has at least one adjacent edge, the edge coverage implies the node coverage. In particular, this holds for connected graphs.

	If $(T,(X_t)_{t \in V(T)})$ is a tree decomposition of $G$, then it is a tree decomposition of all graphs $(V(G), E')$ for $E' \subseteq E$. Furthermore, for $uv \in E$, the tree decomposition $(T, (X'_t)_{t \in T} )$ with $X'_t$ obtained from $X_t$ by replacing each occurrence $u$ or $v$ by $[uv]$, is a tree decomposition of the graph $G^{uv}$. The width of this tree decomposition is at most the width of $(T,(X_t)_{t \in V(T)} )$. Hence a minor of a graph has at most the treewidth of the graph itself. On the other hand, if $(T^{uv}, (X_t)_{t \in V(T^{uv})})$ is a tree decomposition of $G^{uv}$, we can construct a tree decomposition of $G$ by replacing each occurrence of the vertex $[uv]$ in one of the bags by $u$ and $v$. The width of this tree decomposition is at most one larger that the width of $(T^{uv}, (X_t)_{t \in V(T^{uv})})$. Hence shrinking an edge can decrease the treewidth at most by one. The same is true, if we shrink all edges of a matching $M$: If we apply the operation above for each edge of the matching, each node is replaced by two at most once.
	
	If a graph $G$ has several connected components, the treewidth of $G$ is the maximal treewidth of one of its connected components, as we argue next. If $(T^i, (X^i_t)_{t \in V(T^i)} )$ are tree decompositions of the connected components, we can construct a tree decomposition $(T,(X_t)_{t \in V(T)} )$ of $G$ as follows. We start with $V(T)=\bigcup_i V(T^i)$ and $E(T) = \bigcup_i E(T^i)$, i.e.,~the union of all tree decompositions for the connected components. From this forest, we build an arbitrary tree containing it by adding appropriate further edges.
	
	Let $(T, (X_t)_{t \in V(T)})$ be a tree decomposition of a graph $G$ and $tt' \in E(T)$. If $X_t \subseteq X_{t'}$ then $(T^{tt'}, (X'_t)_{t \in V(T^{tt'})})$ with $X'_{\tilde t}=X_{\tilde t}$ for $\tilde t \in V(T^{tt'})  \setminus \{ [tt'] \}$ and $X'([tt'])=X_{t'}$ is a tree decomposition of $G$, i.e., if the bag of a node is a subset of the bag of an adjacent node, we can shrink this node into the adjacent one keeping the larger bag. If there is such a bag in a tree decomposition, we call the tree decomposition redundant and otherwise non-redundant. We argued that there is always a non-redundant tree decomposition of minimal width. The number of nodes of a non-redundant tree decomposition is at most $|V(G)|$. This can be seen by induction on the number of vertices of $G$. If $|V(G)|=1$, this is obvious. Otherwise consider a non-redundant tree decomposition $(T, (X_t)_{t \in V(T)})$ of $G$ and a leaf $t$ of $T$. Let $p(t)$ be the node adjacent to $t$ and $Y=X_t \setminus X_{p(t)}$. As the tree decomposition is non-redundant, $Y \not= \emptyset$. Furthermore, $X_{t'} \cap Y=\emptyset$ for all $t' \not= t$. If we remove $t$ from $T$, we obtain a non-redundant tree decomposition of $G[V(G) \setminus Y]$. By the induction hypothesis, it has at most $|V(G)|-|Y|$ nodes and hence $T$ has at most $|V(G)|$ nodes.
	
	If the treewidth of a graph $G$ is $\tw$, the number of edges $E(G)$ is at most $\tw \cdot |V(G)|$. This can be proven by a very similar induction. Consider again a non-redundant tree decomposition. As above, we obtain a tree decomposition of $G[V(G) \setminus Y]$ by removing a single leaf. By induction hypothesis, the number of edges of $G[V(G) \setminus Y]$ is at most $\tw \cdot (|V(G)-|Y|)$. Each vertex $v \in Y$ can only have neighbors in $X_t$ and hence, its degree is at most $|X_t|-1 \le \tw$. Hence, $G$ has at most $|Y| \cdot \tw$ edges more than $G[V \setminus Y]$.
	
	Let $\tw$ be the treewidth of a graph $G$.
	If $u$ and $v$ have more than $\tw$ common neighbors, there is at least one bag containing $u$ and $v$. Assume otherwise and let $w_1, \dots, w_{\tw +1}$ be common neighbors of $u$ and $v$. First, we can transform the tree decomposition into a non-redundant one without changing this property. There is at least one edge $tt' \in E(T)$ such that $u \notin X_{T^1}$ and $v \notin X_{T^2}$ for the subtrees $T^1$ and $T^2$ of $T$ that arise if one removes $tt'$ from $T$. As $w_i$ is adjacent to $u$, it has to be in at least one bag of $T^1$. Similarly, it has to be in at least one bag of $T^2$. Hence $w_i \in X_t$ and $w_i \in X_{t'}$ for all $1 \le i \le \tw + 1$. As $X_t \not\subseteq X_{t'}$ there is at least one further vertex in $X_t$ and hence $|X_t| \ge \tw+2$, which is a contradiction.

	Choosing an arbitrary node $r \in V(T)$ of a tree $T$ as root, we can make a rooted tree out of $T$ with natural parent-child and ancestor-descendant relations. For a node $t \not= r$ let $p(t)$ be the parent of $t$. Similarly, a tree decomposition $(T, (X_t)_{t \in T})$ can be turned into a rooted tree decomposition by choosing an arbitrary node $r \in V(T)$ as root. For a rooted tree decomposition $(T, (X_t)_{t \in V(t)})$ and $t \in V(T)$, let $X^+_t$ be the set of all vertices in descendants of $t$ and $G^+_t=G[X^+_t]$ the induced graph of these vertices.
	
	A rooted tree decomposition $(T, (X_t)_{t \in T})$ with root $r$ is called nice, if the following properties hold:
	\begin{itemize}
		\item $X_r=\emptyset$ and $X_t=\emptyset$ for all leaves $t$ of $T$, i.e.,~the bags of all leaves and the root are empty
		\item Every non-leaf node $t \in V(t)$ is of one of the following three types:
		\begin{description}
			\item[Join-node:] $t$ has exactly two children $t_1, t_2$ and $X_t=X_{t_1}=X_{t_2}$.
			\item[Introduce-node:] $t$ has exactly one child $t'$ and $X_t=X_{t'} \cup \{ v \}$ for some vertex $v \in V \setminus X_{t'}$. We say that $v$ is introduced at $t$.
			\item[Forget-node:] $t$ has exactly  one child $t'$ and $X_{t'}=X_t \cup \{ v \}$ for some vertex $v \in V \setminus X_t$. We say that $v$ is forgotten at $t$.
		\end{description}
	\end{itemize}
	
	Given a tree decomposition $(T, (X_t)_{t \in T})$ of width $\tw$ of $G$, we can compute in time $\Oh(\tw^2 (|V(T)|+|V(G)|))$ a nice tree decomposition of $G$ of width $\tw$ with at most $\Oh( \tw |V(G)|)$ nodes. This is roughly done as follows: We first make $(T, (X_t)_{t \in T})$ non-redundant and root it at an arbitrary node. Then, we replace each node $t$ with $d \ge 3$ children by a binary tree with $d$ leaves, make each child adjacent to one of the leaves, and set each bag to $X_t$. Now, let $t$ be a node and assume its parent $p(t)$ has degree two. We replace the edge $tp(t)$ by a path of nodes, first forgetting the nodes from $X_t \setminus X_{p(t)}$ one by one and then introducing the vertices $X_{p(t)} \setminus X_t$ one by one.
	
	Notice that if $G$ is connected, the leaves and the root are the only nodes with empty bags if the nice tree decomposition is constructed with the algorithm sketched above. Sometimes it is easier to start the construction with the parents of the leaves whose bags consist of a single vertex.
	
}{ 
	In this Section, we give sketch some definitions of some basic properties, that we will use later in the paper. All definitions are standard and the properties are well known (see, e.g.,~\cite{DBLP:books/sp/CyganFKLMPPS15}).

A tree decomposition $(T, (X_t)_{t \in T})$ for a graph $G$ is a tuple of a tree $T$ over some set of vertices $V(T)$ and subsets of vertices $X_t \subseteq V(G)$ of $G$, one for each vertex in $T$, such that the following three properties hold:
\begin{description}
	\item[(Node coverage)] For each $v \in V(G)$ there is at least one $t \in V(T)$ such that $v \in X_t$.
	\item[(Edge coverage)] For each $uv \in E(G)$ there is at least one $t \in V(T)$ such that $u$ and $v$ are in $X_t$.
	\item[(Coherence)] If $v \in X_{t_1}$ and $v \in X_{t_2}$ for $v \in V(G)$ and $t_1, t_2 \in V(T)$ then $v \in X_{t_3}$ for all vertices $t_3$ on the unique path in $T$ from $t_1$ to $t_2$.
\end{description}

The sets $X_t$ are called the bags.
The width of a tree decomposition $(T, (X_t)_{v \in V(T)})$ is the maximal cardinality of one of the sets $X_t$ minus one, i.e., $\max_{t \in V(T)} |X_t|-1$. The treewidth of a graph $G$ is the minimal width of a tree decomposition of $G$.

To make a clearer distinction between the vertices $V(G)$ of the graph and those of the tree of the tree decomposition, we call the elements of $V(T)$ nodes and the elements of $V(G)$ vertices in the following. 

If there are bags $t \not= t'$ in a tree decomposition with $X_t \subseteq X_{t'}$, we call the tree decomposition redundant and otherwise non-redundant. There is always a non-redundant tree decomposition of minimal width. The number of nodes of a non-redundant tree decomposition is at most $|V(G)|$. 

Choosing an arbitrary node $r \in V(T)$ of $(T,(X_t)_{t \in T})$ as root, we get a rooted tree decomposition with natural parent-child and ancestor-descendant relations. For a node $t \not= r$ let $p(t)$ be the parent of $t$. For a rooted tree decomposition $(T, (X_t)_{t \in V(t)})$ and $t \in V(T)$, let $X^+_t$ be the set of all vertices in descendants of $t$ and $G^+_t=G[X^+_t]$ the induced graph of these vertices.

A rooted tree decomposition $(T, (X_t)_{t \in T})$ with root $r$ is called nice, if the following properties hold:
\begin{itemize}
	\item $X_r=\emptyset$ and $X_t=\emptyset$ for all leaves $t$ of $T$, i.e.,~the bags of all leaves and the root are empty
	\item Every non-leaf node $t \in V(t)$ is of one of the following three types:
	\begin{description}
		\item[Join-node:] $t$ has exactly two children $t_1, t_2$ and $X_t=X_{t_1}=X_{t_2}$.
		\item[Introduce-node:] $t$ has exactly one child $t'$ and $X_t=X_{t'} \cup \{ v \}$ for some vertex $v \in V \setminus X_{t'}$. We say that $v$ is introduced at $t$.
		\item[Forget-node:] $t$ has exactly  one child $t'$ and $X_{t'}=X_t \cup \{ v \}$ for some vertex $v \in V \setminus X_t$. We say that $v$ is forgotten at $t$.
	\end{description}
\end{itemize}

Given a tree decomposition $(T, (X_t)_{t \in T})$ of width $\tw$ of $G$, we can compute in time $\Oh(\tw^2 (|V(T)|+|V(G)|))$ a nice tree decomposition of $G$ of width $\tw$ with at most $\Oh( \tw |V(G)|)$ nodes. 
}
	
	\section{Computing an Optimal Tree Decomposition from an Arbitrary One}\label{KloksAlg}
	
	We want to compute an optimal tree decomposition of a graph $G=(V,E)$. Notice that there is an optimal tree decomposition consisting of at most $n$ nodes. 
	There are only a finite number of tree topologies with at most $n$ nodes. Hence we can enumerate all such topologies and all assignments of subsets of $V$ with a size of at most $\tw$ to the nodes and check the three properties of a tree decomposition to compute the optimal tree decomposition. We now want to improve upon this simple algorithm by using a given tree decomposition of $G$.
	
	\subsection{Enumerating Tree Decompositions with a Detour}

	Let $G=(V,E)$ be the given graph and $\given\td=(\given T,(Y_{\given t})_{\given t \in V(\given T)})$ be a nice (and hence rooted) tree decomposition of $G$ of width $\given\tw$. Both the graph $G$ and the tree decomposition $\given\td$ are given to the algorithm and we fix this notation for the remainder of this section. We assume that $G$ is connected and hence edge coverage implies node coverage. 
	
	\begin{figure}
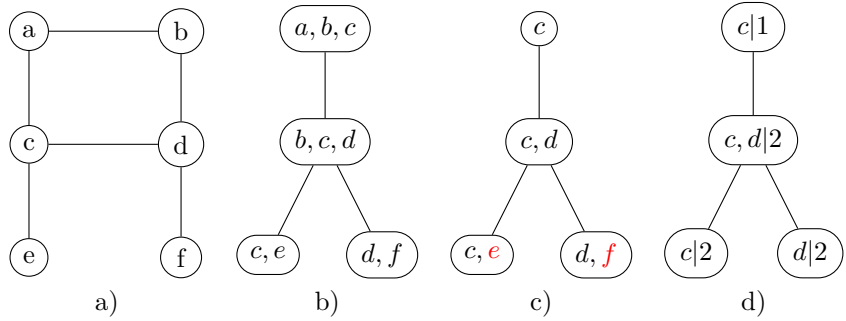

		\begin{center}
			\scalebox{\figscale}{
		\ls{\begin{tabular}{cccc}
			\graphg & 
			\tdg &
			\tdr &
			\tdrb \\
			a) & b) & c) & d)
		\end{tabular}
	}{
		\begin{tabular}{ccccc}
		\graphg & 
		\tdg &
		\tdr &
		\ceinss &
		\czweis \\
		a) & b) & c) & d) & e)
	\end{tabular}
	}
	}
		\end{center}
	\ls{}{\figspace}
		\caption{Consider the graph on the left, (a). (b) shows an optimal tree decomposition for it. If we remove the vertices $a$ and $b$ from all bags as shown in (c), we obtain a tree decomposition for $G[\{c, d, e, f \}]$. Assume that the given tree decomposition has a node $\given t$ with $X_{\given t}=\{c,d\}$ and $X^+_{\given t}=\{c, d, e, f\}$. Then the tree decompositions enumerated for $\given t$ already covered all edges incident to $e$ and $f$ and hence, it does not matter whether a bag contains these vertices, only the current size of the bag matters. Formally, this is captured by the restricted bags depicted in (d)\ls{.}{ ignoring the green part. We can assume that no node is added to the bottom-right leaf in (d) later in the construction, as its restricted bag is contained in the restricted bag of its parent and contains additional nodes. This is as we could add an additional leave only containing the nodes of the restricted bag and the additional nodes as depicted in (e).}  }\label{treedecomposition}
	\end{figure}
		
	In order to reduce the number of enumerated tree decompositions, we want to make use of the given (non-optimal) tree decomposition. Notice that given a tree decomposition $\td$ for $G$, we obtain a tree decomposition for $G^+_{\given t}$ if we intersect all bags of $\td$ with $Y^+_{\given t}$\ls{( see Figure \ref{treedecomposition}, part a), b) and c))}{}. For all $\given t \in V(\given T)$, we enumerate all tree decompositions of $G^+_{\given t}$ of width at most $\given\tw$ with at most $n$ nodes bottom up, denoted as $\TD_{\given t}$. Hence, for the root $\given r$ of the given tree decomposition, we enumerate all tree decompositions of width at most $\given\tw$ of $G$ and the problem is solved. Moreover, it is easy to enumerate $\TD_{\given t}$ from its children as follows:

\begin{construction}(Base Case)
\begin{description}
	\item[Leaves:] For a leaf $\given t$ of $\given \td$, the set $\TD_{\given t}$ contains all tree decompositions $(T, (\emptyset)_{t \in V(T)})$ for an arbitrary tree $T$ of at most $n$ nodes.
	\item[Join-Node:] Consider an arbitrary tree decomposition $\td=(T, (X_t)_{t \in V(T)})$ in $\TD_{\given t}$ for a join-node $\given t$ with children $\given t^1$ and $\given t^2$. Notice that $\td^i=(T, (X_t \cap Y^+_{\given t^i})_{t \in V(t)}) \in \TD_{\given t^i}$ for $i\in \{1,2\}$ hence enumerated in $\TD_{\given t^i}$. 
	Furthermore the tree $T$ and $X_t \cap Y_{\given t}$ for all $t \in T$ are the same on all three tree decompositions. Therefore all tree decompositions in $\TD_{\given t}$ can be constructed by choosing tree decompositions $\td^1=(T^1, (X^1_t)_{t \in V(T^1)}) \in \TD_{\given t^1}$ and $\td^2=(T^2, (X^2_t)_{t \in V(T^2)}) \in \TD_{\given t^2}$ with the same tree, i.e.,~$T^1=T^2$ and such that $X^1_t \cap Y_{\given t} = X^2_t \cap Y_{\given t}$ for all $t \in T^1$ and constructing $\td=(T^1, (X^1_t \cup X^2_t)_{t \in V(T^1)})$. For a join-node $\given t$, we define $J_{\given t}: \TD_{\given t^1} \times \TD_{\given t^2} \rightarrow 2^{\TD_{\given t}}$ by $J_{\given t}(\td^1, \td^2)= \{ \td \}$ for $\td^1, \td^2$ and $\td$ as above, i.e.,~$J_{\given t}(\td^1, \td^2)$ is the set containing the tree decomposition constructed from $\td^1$ and $\td^2$ as explained above as the single element. $J_{\given t}$ maps to the power set of all tree decompositions as this is consistent with the corresponding definition for introduce- and forget-nodes.
	\item[Introduce-Node:] Consider an arbitrary tree decomposition $\td=(T, (X_t)_{t \in V(T)})$ in $\TD_{\given t}$ for an introduce-node $\given t$ with child $\given t'$ for which the vertex $v \in V$ is introduced. Notice that $\td'=(T, (X_t \setminus \{v\})_{t \in V(T)})$ is a tree decomposition in $\TD_{\given t'}$. Furthermore the nodes $t$ of $V(T)$ such that the bag $X_t$ contains $v$ form a subtree of $T$.  For each edge $vw \in E(G^+_{\given t})$ adjacent to $v$, there is at least one bag containing $v$ and $w$ in $T$. Hence, we can enumerate all tree decompositions in $\TD_{\given t}$ as follows: We choose a tree decomposition $\td'$ in $\TD_{\given t'}$ and a subtree of the tree of $\td'$ with the properties mentioned earlier. For these choices, we construct the tree decomposition in $\TD_{\given t}$ from $\td'$ by adding $v$ into the bags of the selected subtree. We refer to the set that contains all these tree decompositions as $I_{\given t}(\td')$.
	\item[Forget-Node:] Consider an arbitrary tree decomposition $\td=(T, (X_t)_{t \in V(T)})$ in $\TD_{\given t}$ for a forget-node $\given t$ with child $\given t'$ for which the vertex $v \in V$ is forgotten. Notice that $\td$ is also a tree decomposition in $\TD_{\given t'}$ (as $G^+_{\given t}=G^+_{\given t'}$) and hence the tree decomposition in $\TD_{\given t}$ are the same as the ones in $\TD_{\given t'}$. We define $F_{\given t}:\TD_{\given t'} \rightarrow 2^{\TD_{\given t}}$ by $F_{\given t}(\td)= \{\td \}$.
\end{description}
\end{construction}

	We do not argue about the running time to enumerate these tree decompositions (since we will be computing something different anyway), but we notice that the number of enumerated tree decompositions is at most $T_n \cdot n^{n \cdot \given \tw}$, where $T_n$ is the number of tree topologies with at most $n$ nodes.
	
	This estimate of the number of enumerated tree decompositions is a function in $n$ and $\given\tw$ and not only in $\given\tw$. We now reduce the number of enumerated tree decompositions to make it a function in $\given \tw$ only. We show that it suffices to store some limited information for a tree decomposition $\td$ that is enumerated at a node $\given t$ of the given tree decomposition $\given \td$: We store only a part of the tree of $\td$ and only those vertices in the bags of $\td$ that are in $Y_{\given t}$ of $\given \td$. Hence, several tree decompositions have the same limited information. We write $\td^1 \equiv^{\given t}_s \td^2$ if two tree decompositions have the same limited information and call them equivalent. Before defining $\equiv^{\given t}_s$, we will define finer classes $\equiv^{\given t}_b$ and $\equiv^{\given t}_c$, i.e., we iteratively show that we can ignore some information in the construction of the tree decompositions.
	
	The classes will be such that equivalent tree decompositions have the same treewidth and if one can use one tree decomposition to construct a certain tree decomposition in the parent node of the given tree decomposition, we can use each equivalently to construct an equivalent one. In order do compensate for the removal of nodes of the tree, we allow to add additional leaves whose bags are subsets of the parental bags and copy certain nodes of degree two. The set of all tree decompositions that can be constructed from the limited information of $\td$ is denoted by $\E^{\given t}(\L^{\given t}(\td))$.  More formally, the following will hold for the equivalence class $\equiv^{\given t}_s$ and similarly for the classes $\equiv^{\given t}_b$ without possible modifications of the tree decomposition $\td$ and $\equiv^{\given t}_c$ by allowing to add additional leaves to $\td$.
	
	\begin{lemma}\label{cancompress} \
	\begin{description}
		\item[Join-node:] Let $\given t$ be a join-node with child nodes $\given t^1$ and $\given t^2$, $A, A' \in \TD_{\given t^1}$ with $A \equiv^{\given t^1}_{b} A'$ and $B, B' \in \TD_{\given t^2}$ with $B \equiv^{\given t^2}_{b} B'$. If $C \in J_{\given t}(E^{\given t}(\L^{\given t}(A)),E^{\given t}(\L^{\given t}(B)))$ then there is $C' \in J_{\given t}(E^{\given t}(\L^{\given t}(A')),E^{\given t}(\L^{\given t}(B')))$ with $C' \equiv^{\given t}_{b} C$.
		\item[Introduce-node:] Let $\given t$ be an introduce-node with child node $\given t'$ and $A, A' \in \TD_{\given t'}$ with $A \equiv^{\given t'}_{b} A'$. If $C \in I_{\given t}(E^{\given t}(\L^{\given t}(A)))$ then there is $C' \in I_{\given t}(E^{\given t}(\L^{\given t}(A')))$ with $C' \equiv^{\given t}_{b} C$.
		\item[Forget-node:] Let $\given t$ be a forget-node with child node $\given t'$ and $A, A' \in \TD_{\given t'}$ with $A \equiv^{\given t'}_{b} A'$. If $C \in F_{\given t}(E^{\given t}(\L^{\given t}(A)))$ then there is $C' \in F_{\given t}(E^{\given t}(\L^{\given t}(A')))$ with $C' \equiv^{\given t}_{b} C$.
	\end{description}
	\end{lemma}

	\begin{proof}
	As we will see, the limited information of the tree decompositions in $J_{\given t}(E^{\given t}(\L^{\given t}(A)),E^{\given t}(\L^{\given t}(B)))$, $I_{\given t}(E^{\given t}(\L^{\given t}(A)))$ and  $F_{\given t}(E^{\given t}(\L^{\given t}(A)))$ can be computed from the limited information of $A$ (and $B$). Therefore, if we replace $A$ by $A'$ (and $B$ by $B'$) in the construction, the construction will lead to tree decompositions with the same limited information.
	\end{proof}
		
	\subsection{Equivalence classes on bags}
	
	
	Given a tree decomposition $\td=(T, \{X_t\}_{t \in V(T^1)})$, we define the restricted bags with respect to $\given t$ to be the tuples $(X_t \cap Y_{\given t}, |X_t|)_{t \in V(T)}$, i.e.,~the vertices of the bag are restricted to the vertices of the current bag of the given tree decomposition and we store the sizes of the bags\ls{ (see Figure \ref{treedecomposition}, part d))}{}. The restricted bag representation of a tree decomposition $\td=(T,(X_t)_{t \in V(T)})$ consists of the tree $T$ and of the restricted bags for all $t \in V(T)$.
	We call tree decompositions $\td^1=(T^1, \{X^1_t\}_{t \in V(T^1)})$ and $\td^2=(T^2, \{X^2_t\}_{t \in V(T^2)})$ in $\TD_{\given t}$ bag-equivalent, if their restricted bag representations are the same and write $\td^1 \equiv^{\given t}_{b} \td^2$ in his case. \ls{We will now show how we can compute the restricted bag representation of a constructed tree decomposition directly from the restricted bag representation of the tree decompositions used in the construction. We will then conclude that it suffices to enumerate one tree decomposition per equivalence class.}{}

	\ls{
	\begin{construction} (Restricted Bags)
		\begin{description}
		\item[Join-node:]
		Let $A=(T^A, (X^A_t)_{t \in V(T^A)}) \in \td^1$ and $B=(T^B, (X^B_t)_{t \in V(T^B)}) \in \td^2$ tree decompositions with $T^A=T^B$ and $X^A_t \cap Y_{\given t}=X^B_t \cap Y_{\given t}$ for all $t \in V(T^A)$. The (unique) tree decomposition $C \in J(A,B)$ has again the same tree and the bags are just the union of the corresponding bags in $A$ and $B$. The restricted bag of a node $t$ of $C$ is $(X^A_t \cup X^B_t) \cap Y_{\given t}=(X^A_t \cap Y_{\given t}) \cup (X^B_t \cap Y_{\given t}) =X^A_t \cap Y_{\given t}$ and its size is the sum of the sizes of the restricted bags of $t$ in $A$ and $B$ minus $|X^A_t \cap Y_{\given t}|$ (as $X^A_t \cap Y_{\given t}$ is the set of nodes shared by $X^A_t$ and $X^B_t$). Notice that the set of vertices of the restricted bags of $C$ and its sizes can be computed directly from the restricted bags of $A$ and $B$. 
		\item[Introduce-node:] Assume $C$ is constructed from $A$ if the introduced vertex $v$ is added to the bags of the subtree $T$ of the tree of $A$. 
		$T$ can be used, if all edges adjacent to $v$ of $G^+_{\given t}$ are covered. As the other endpoints of these edges have to be in $Y_{\given t}$, we can decide whether $T$ is valid from the restricted bags. The restricted bags of $C$ are the same as those of $A$, but $v$ is added to the bags of the nodes in $T$. The sizes of the bags of nodes in $T$ increase by one, and the other sizes do not change. Hence, we can compute the restricted bag representation of $C$ directly from the restricted bag representation of $A$.
		\item[Forget-node:] As the unique tree decomposition $C \in F(A)$ is $A$, the restricted bag representation of $C$ is the restricted bag representation of $A$ except that we have to delete the vertex that is forgotten from each restricted bag.
	\end{description}	
	\end{construction}
}{We will not show the construction of $J^{\given t}$, $I^{\given t}$ and $F^{\given t}$ but give the main points. For a join-node, the vertices of the restricted bags have to be the same and the size can easily be computed out of the sizes of the children. For an introduce-node, we note that the validity of a chosen subtree for the new vertex only depends on the vertices that are contained in the restricted bags of the given tree decomposition. For a forget-node, we simply remove the forgotten vertex from all restricted bags.}

	\ls{Notice that the restricted bag representation of $C \in J_{\given t}(A,B) \cup I_{\given t}(A) \cup F_{\given t}(A)$ can be directly computed from the restricted bag representation of $A$ (and $B$) and hence Lemma \ref{cancompress} holds for $\equiv^{\given t}_b$.}{}
	
	
	\subsection{The Core of a tree decomposition}
	
	In this section, we will use the following fact about tree decompositions: If $t$ is a node with bag $X_t$, we can add a new node $t'$, the edge $tt'$ and chose $X_{t'} \subseteq X_t$ to obtain an other tree decomposition of the same width. We call this an addition of a leaf to the tree decomposition.
	
	In the following, we will show that we can assume some structure of the constructed optimal tree decomposition which depends on the given tree decomposition $\given\td=(\given T, (Y_{\given t})_{\given t \in V(\given T)})$. 
	
	\begin{lemma}\label{core}
	Assume $t$ is a leaf with adjacent node $p(t)$ in the constructed tree decomposition $\td=(T, (X_t)_{t \in V(t)})$ such that $X_t \setminus Y_{\given t} \not= \emptyset$ and $X_t \cap Y_{\given t} \subseteq X_{p(t)} \cap Y_{\given t}$ for some node $\given t \in V(\given T)$. Then there is an optimal tree decomposition such that no node from $V \setminus Y^+_{\given t}$ is contained in $X_t$.		
	\end{lemma}

	\begin{proof}	
 	Let the notation be given as stated in the lemma. Assume $Z:= X_t \cap (V \setminus Y^+_{\given t})$ is non-empty. We can remove all vertices in $Z$ from $X_t$ and create a new child of $p(t)$ with bag $Z \cup (X_t \cap Y^+_{\given t})$ to get another tree decomposition\ls{ (see Figure \ref{lemmacore} for an illustration)}{}. This transformation cannot increase with the width of the tree decomposition and all edges are still covered. 
 	\end{proof}
	
	\ls{
	\begin{figure}
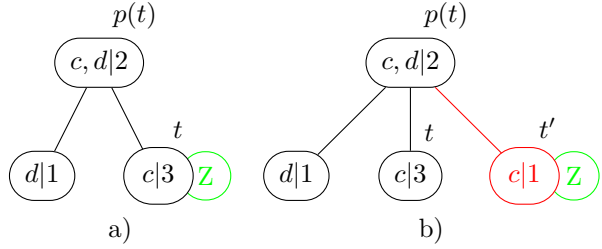

		\begin{center}
			\scalebox{\figscale}{
		\begin{tabular}{cc}
			\ceins & \czwei\\
			a) & b)
		\end{tabular}}
		\end{center}
	\ls{}{\figspace}
		\caption{Assume the leaf $t$ of the tree decomposition shown in (a) contains only vertices in the restricted bag that are also contained in that of its adjacent node $p(t)$ ($\{ c \}$ in this case), but $X_t$ contains further vertices from $Y^+_{\given t} \setminus Y_{\given t}$, i.e., vertices that are already forgotten. In the proof of Lemma \ref{core}, we showed that we may assume that $X_t$ does not get vertices from $V \setminus Y^+_{\given t}$ (the green nodes $Z$ in the figure), as we could create another node $t'$ (shown in red in (b)) that can get these additional nodes and the nodes in the restricted bag. Both nodes $t$ and $t'$ now contain less vertices as the node $t$ in (a)). }\label{lemmacore}
	\end{figure}
}{}
	
	Informally, a leaf with a bag whose vertices are a subset of the vertices of its parent gets only  additional vertices later on in the construction if it does not contain vertices that are already forgotten. Notice that if a node becomes a leaf by the removal of a node, the rule can be applied to this node too.
	
	
	
	Given a tree decomposition $\td=(T, \{X_t\}_{t \in V(T)}) \in \TD_{\given t}$ we define its $\given t$-core as follows: As long as there is a leaf $t$ in $T$ such that $X_t \cap Y_{\given t} \subseteq X_{p(t)} \cap Y_{\given t}$ remove the leaf from $t$ where $p(t)$ is the node adjacent to $t$ (see Figure \ref{figurecore}, part a) and b)). Since it is possible to remove the largest bag in this way, we store the size of the largest bag with the core. We call two tree decompositions $\td^1$ and $\td^2$ ${\given t}$-core-equivalent if their cores are identical (including the restricted bags of the nodes that were not removed from the core and the size of the largest bag), denoted by $\td^1 \equiv^{\given t}_{c} \td^2$. In the following, it is easier if we can assume that $Y_{\given t} \not= \emptyset$ and hence, we start the construction with the parents $p(\given t)$ of the leaves $\given t$. Their cores consist of a single node whose bag contains the single vertex in $Y_{p(\given t)}$.

	In order to compensate this removal of bags, we have to allow the adding of additional leaves in the construction. By Lemma \ref{core}, we can assume that the bag of a leaf connected to a node $t$ that contains a node from $V \setminus Y^+_{\given t}$ does not contain nodes from $Y^+_{\given t} \setminus X_t$. For a tree decomposition $\td$ let $\L^{\given t}(\td)$ be the (infinite) set of all tree decompositions that can be obtained by iteratively adding leaves to nodes of the $\given t$-core of $\td$, whose bags are subsets of $Y_{\given t}$. $\L^{\given t}$ depends on the node $\given t$ of the given tree decomposition as we have to start the addition on a node of the core. 
	Notice that all elements of $\L^{\given t}(\td)$ have the same core as $\td$.
	\ls{Let $J^{\L}_{\given t}(\td^1, \td^2)=J_{\given t}(\L^{\given t}(\td^1), \L^{\given t}(\td^2))$, $I^{\L}_{\given t}(\td)=I_{\given t}(\L^{\given t}(\td))$ and $F^{\L}_{\given t}(\td)=F_{\given t}(\L^{\given t}(\td))$.}{}
	
	Notice that each leaf of the core has a unique vertex as otherwise each vertex of the leaf would also be contained in its adjacent node and hence the leaf can be removed. Hence, the number of leaves is at most $\given\tw$. Furthermore, the number of nodes of the core of degree at least three is at most $\given\tw$ as the tree has at most $\given\tw$ leaves. In the following, we give up assumption that the constructed tree decomposition has at most $n$ nodes, but we assume that each path of nodes of degree $2$ in the core has a length of at most $n$.

	\begin{figure}
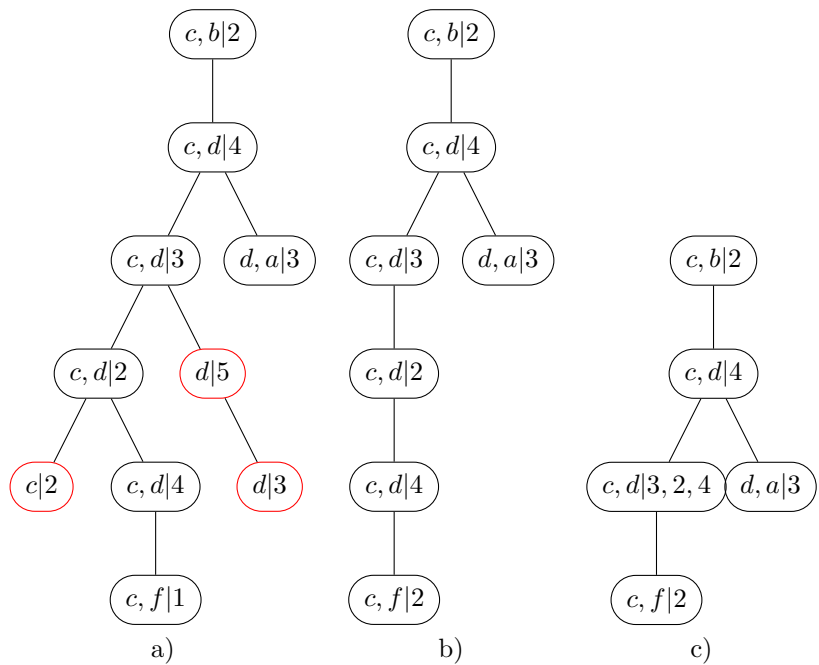

		\begin{center}
			\scalebox{\figscale}{
		\begin{tabular}{ccc}
		\tdzweirb & \tdzweic & \tdzweicc \\
			a) & b) & c)
		\end{tabular}}
		\end{center}
		\ls{}{\figspace}
		\caption{Assume we consider a tree with its restricted bags as shown in (a). The nodes shown in red are removed from the core (depicted in (b)) as they are leaves and their restricted bags are subsets of the bags of their adjacent nodes. Notice that we remove the largest bag and hence we store its size with the core. In (c) we show the compact representation of the core, i.e.,~the path of nodes of degree $2$ having equal vertices in the restricted bag are replaced by a single node and the sizes of the original nodes are represented as an integer seqeunce in the replacement.}\label{figurecore}
	\end{figure}
	
	\ls{
	Before we show that only one tree decomposition per equivalence class suffices, we will show how we have to adopt the construction to compute $J^\L_{\given t}, I^\L_{\given t}$ and $F^\L_{\given t}$.}{}
	
	\ls{
	\begin{construction}(Construction for all tree decompositions in $\L^{\given t}(\td)$)
		\begin{description}
			\item[Join-node:] Let $A \in \TD_{\given t^1}$ and $B \in \TD_{\given t^2}$. Notice that there are $A^{\L} \in \L^{\given t}(A)$ and $B^{\L} \in \L^{\given t}(B)$ such that their trees are the same (and hence $J(A^{\L}, B^{\L})$ is non-empty) if and only if their cores are the same. Furthermore, the core of the constructed tree decomposition $C$, i.e.,~the unique element of $J(A^{\L},B^{\L})$, has the same core as $A$ and $B$. The size of the largest bag of $C$ is the maximum of the largest bag in the core, the maximum in $A$, and the maximum in $B$. Hence all such $J(\L^{\given t}(A), \L^{\given t}(B))$ are core-equivalent and can be constructed directly from the cores of $A$ and $B$.
			
			\item[Introduce-node:] Let $A \in \TD_{\given t'}$ and consider an $A^{\L} \in \L^{\given t}(A)$ which is used when introducing the vertex $v$ and let $C$ be the constructed tree decomposition. If $v$ is contained in at least one bag of the core of $A^{\L}$, the core of $C$ is exactly the same as the one of $A$ (as we can apply the rule to a node before the addition of $v$ if and only if we can apply it after the addition of $v$). This is illustrated in Figure \ref{figurecore}, part a). Hence, the core of $C$ can be constructed by selecting a subtree of the core of $A$ and adding $v$ to the bags of the nodes of the subtree.
			
			Otherwise, a new path is added to the core of $A$ at an arbitrary node $t$ with bags that are subsets of $X_t$ such that the path is consistent and the bag of the last node of the path contains $v$ (see Figure \ref{figurecore}, part b)). Hence, for the construction, we have to consider all nodes $t$ of the core and add an arbitrary path (of length at most $n$) to $t$. Then we assign a nested sequence of subsets of $X_t$ to the nodes of the path and finally add $v$ to the bag of the last node of the path.
			
			\item[Forget-node:] For each $A \in \TD_{\given t'}$ and $A^\L \in \L^{\given t}(A)$, the single element of $F_{\given t}(A^L)$ is $A^L$. Its $\given t$-core can be obtained by starting with its $\given t'$-core (which is the $\given t'$-core of $A$) and potentially removing further leaves since the bags are now intersected with a smaller set.
		\end{description}
	\end{construction}
}{We will not show the construction for $J^{\given t}(\L^{\given t}(A), \L^{\given t}(B)), I^{\given t}(\L^{\given t}(A))$ and $F^{\given t}(\L^{\given t}(A))$, but give the main points. For a join-node, we note that if we can join two tree decompositions, their cores must be the same. For an introduce-node we have to distinguish two cases. If the vertex is added to at least one bag of the core, the core does not change. Otherwise the new vertex will be in exactly one bag of the new core, which is connected to the old core by a single path which we removed then constructing the core. Hence, we have to try all possibilities of adding path of nested subsets of length at most $n$. For a forget-node, we potentially have to remove further nodes from the core.}

\ls{
	\begin{figure}
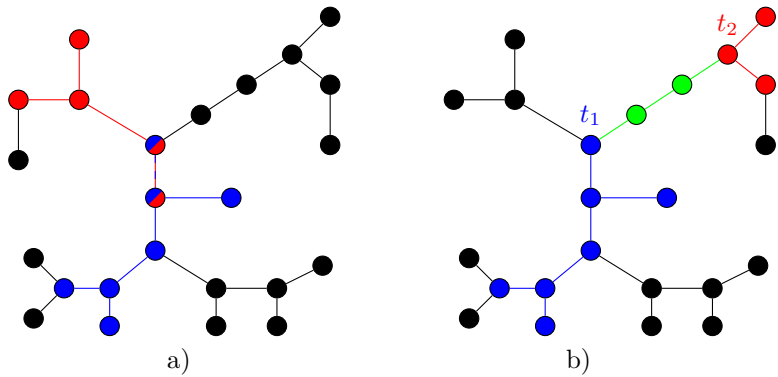

		\begin{center}
			\scalebox{\figscale}{
		\begin{tabular}{cc}
			\topologiezwei \qquad &  \qquad \topologie \\
			a) & b)
		\end{tabular}}
	\end{center}
	\ls{}{\figspace}
	\caption{The underlying tree is shown in black, its core before the vertex was introduced (the old core) is shown in blue and the vertex that is inserted into the bags of the nodes is shown in red. In part a), the vertex is inserted in at least one node of the old core. All black nodes and all red nodes that were not in the old core can still be removed and hence the new core is the old one. In part b), the vertex is inserted only into bags of nodes that were not in the old core. Let $t_1$ be the last blue node and $t_2$ the first red node on the path from the blue tree to the red tree, shown in green. All red nodes but $t_2$, and all black nodes but those on the path from $t_1$ to $t_2$ can be removed and hence are not in the new core. Hence, the new core consists of the blue nodes, the green nodes and $t_2$. The bags of the green nodes and the bag of $t_2$ without the introduced vertex are a sequence of nested subsets of the bag of $t_1$. }\label{introcore}
	\end{figure}
	
}{}

	\ls{Notice that the core of $C$ can be directly computed from the core of $A$ (and $B$) and hence Lemma \ref{cancompress} holds for $\equiv^{\given t}_c$.}{}
	
	\subsection{The compressed core of a tree decomposition}
	
	In the following, it is easier to consider a different representation of the cores. We replace each maximal path of nodes in the tree of degree $2$ such that the vertices in the restricted bags are the same by a single node, and assign the sequence of integers of the sizes of the bags of the path with the new node (see Figure \ref{figurecore}, part c)). We call this the compact representation of the core. Notice that in this representation, the number of nodes becomes bounded in $\given\tw$ since between two nodes $t$ and $t'$ of degree at least three, we have at most $|X_t \setminus X_{t'}|+|X_{t'} \setminus X_t| \le 2\given\tw$ nodes of degree two. On the other hand, each node has an integer sequence assigned to it of length up to $n$. 
	We define equivalence classes on integer sequences such that the number of equivalence classes becomes bounded in $\given\tw$. We call two tree decompositions $\td^1$ and $\td^2$ sequence-equivalent, if the compact representations of their cores are the same and each pair of integer sequences assigned to a node in the compact representation of the core are equivalent, denoted as $\td^1 \equiv^{\given t}_s \td^2$.

	In this section, we use the following fact about tree decompositions: Assume that $\td=(T, (X_t)_{t \in V(T)})$ is a tree decomposition of $G$ and $tt'$ is an edge of $T$. If we add a bag $t''$ with $X_{t''}=X_t$ to $T$, remove the edge $tt'$ and add edges $tt''$ and $t''t'$, we get another tree decomposition of $G$ of the same width. We call the resulting tree decomposition an extension of $\td$. 
	In the following, we restrict this operation to edges $tt'$ such that the degree of $t$ is two in the core of $\td$. For a tree decomposition $\td$ let $\E^{\given t}(\td)$ be the set of all tree decomposition that can be obtained from $\td$ by a series of such extensions. 
	Notice that although $\E^{\given t}(\td)$ is infinite, we will define a finite number of equivalence classes. Notice that the tree of the compact representation of the core is the same for each element of $\E^{\given t}(\L^{\given t}(\td))$, but the integer sequences differ.
	
	In the next lemma we show that we may assume that a vertex added to a subpath in $\td$ of $t_1, \dots, t_k$ with integer sequence $(a_1, \dots, a_k)$ is added up to a node $t_\ell$ such that there are no $i \le \ell \le j$ with $\min(a_i,a_j) < a_\ell < \max(a_i,a_j)$.
		
	\begin{lemma}\label{compact}
		Consider a node $\given t$ of the given tree decomposition and the compact representation of the core of a tree decomposition in $\TD_{\given t}$. 
		Assume $(a_1, \dots, a_k)$ is the integer sequence of a node and $i < j$ are such that $a_i \ge a_k$ and $a_j \le a_k$ for all $i \le k \le j$. Let $\ell$ be the smallest index such that $a_\ell > a_i$ (if such an $\ell$ exists, set $\ell=j$ otherwise). Let $Z_i, \dots, Z_j$ be the vertices of $V \setminus Y^+_{\given t}$ of bags of the nodes corresponding to $a_i, \dots, a_j$ in the final tree decomposition. There is an optimal tree decomposition such that $Z_{\ell'}=Z_j$ for all $\ell \le \ell' \le j$ and no edges adjacent to $a_{\ell'}$ are added to the core later in the construction. 
		Similarly, if $a_i \le a_k$ and $a_j \ge a_k$, we can assume $Z_{\ell'}=Z_i$ and no edges adjacent to $a_{\ell'}$ are added to core for all $i \le \ell' \le \ell$, for $\ell$ being the largest index such that $a_\ell > a_j$.
		
	\end{lemma}
	
	\begin{figure}
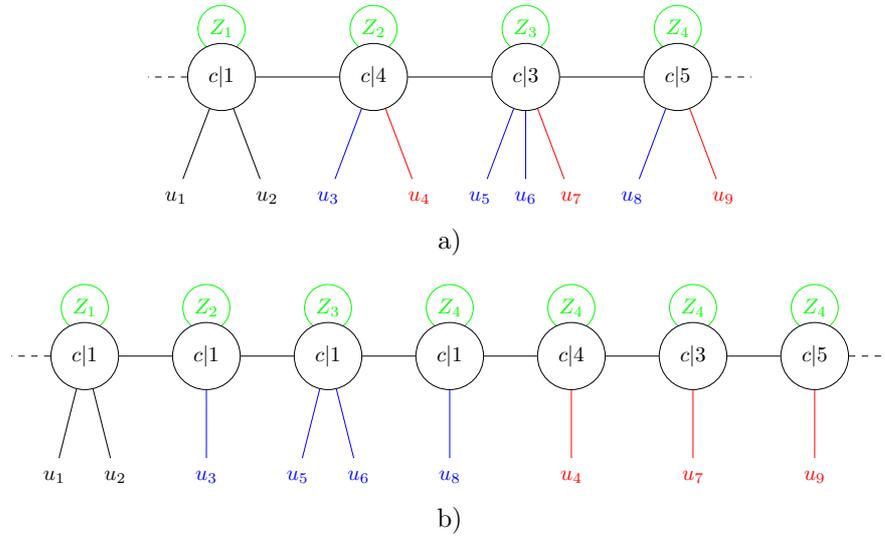

		\begin{center}
			\scalebox{\figscale}{
		\ls{\begin{tabular}{c}
			\scalebox{0.8}{\cceins} \\
			a)\\[1em]
			\scalebox{0.8}{\cczwei}\\
			b)
		\end{tabular}}{
		\begin{tabular}{cc}
	a) & \cceins \\
	b) & \cczwei
\end{tabular}}}
		\end{center}
		\ls{}{\figspace}\vspace*{-1em}
		\caption{Assume we have a path of four nodes in the core such that the sets of vertices of the restricted bags are equal (and hence these and possible further nodes are replaced by a single node in the restricted core), say $X$. Assume furthermore that the first node has the smallest number of vertices and the last has the largest number, and that the sets of vertices $Z_1, \dots, Z_4$ are added to these nodes later on in the construction. We proved that we may assume that $Z_2=Z_3=Z_4$. Otherwise, we could add three copies of the first node in the beginning and give the first three nodes the additional vertices $Z_1, Z_2$ and $Z_3$ and all remaining nodes the additional vertices $Z_4$. By Lemma \ref{core}, the bags of nodes incident to the nodes may be assumed to be either subsets of $X \cup Y^+_{\given t}$ (shown in red) or subsets of $X \cup (V \setminus Y^+_{\given t})$ (shown in blue). The first kind of edges stay at their original positions whereas the second kind are moved to the front.}\label{figuretypcial}\vspace*{-1em}
	\end{figure}

	\begin{proof}
		Consider an arbitrary tree decomposition $\td$ and assume it does not have the property. Let $t$ be a node of the compact representation of the core for which the property does not hold for the subsequence $(a_i, \dots, a_j)$ of its integer sequence and let $t_i, \dots, t_j$ the nodes of the tree of $\td$ corresponding to $a_i, \dots, a_j$.
		Notice that we can assume that the edges adjacent to the nodes corresponding to $t_i, \dots, t_j$ have bags that are either subsets of $Y^+_{\given t} \cap X_t$ or subsets of $(V \setminus Y^+_{\given t}) \cup X_t$ by Lemma \ref{core}.
		
		We extend $t_i$, i.e.,~the node with the smallest bag, $j-i+1$ times, i.e., instead of a single node $t_i$ for size $a_i$, we now have nodes $t_i^1, \dots, t_i^{j-i+1}$, all of size $a_i$. We move the sets $Z_i, \dots, Z_j$ to the nodes $t_i^1, \dots, t_i^{j-i+1}$ and with them all edges whose other incident nodes have bags are subsets of $(V \setminus Y^+_{\given t}) \cup X_t$. The bags of nodes $t_2, \dots, t_j$ are all set to $Z_j$ (see Figure \ref{figuretypcial} for an illustration). 
		
		This results in a tree decomposition of at most the same width (The bags of $t_i$ and $t_j$ are not changed. For the new nodes $t_i^k$, we have that $|X_{t_i^k}|$ at most the size of bag of the node that contained the assigned $Z_{\ell}$ in the originating tree decomposition. The bags of the nodes $t_{i+1}, \dots, t_{j-1}$ have at most the size of that of $t_j$) that has the property for $a_i, \dots, a_j$. We can repeat this construction until all such conditions are satisfied.
	\end{proof}

	Notice furthermore that if we add a vertex $v$ up to the node $t_\ell$, we can first extend $t_\ell$ and add only up to the first appearance of $t_\ell$ such that after the insertion, the first node not containing $v$ is the copy of $t_\ell$ and still has size $a_\ell$.

	For an integer sequence $A=(a_1, \dots, a_k)$, we define its typical sequence $\tau(A)$ as follows: Apply one of the following two operations until no further such operation is possible (we will show that the resulting sequence is unique and hence it is indeed a definition)
	\begin{itemize}
		\item if $a_i=a_{i+1}$ for $1 \le i \le k-1$, replace $(a_1, \dots, a_k)$ by $(a_1, \dots, a_i, a_{i+2}, \dots, a_k)$
		\item 	for $1 \le i<j \le k$ with $min(a_i,a_j) \le a_\ell \le max(a_i,a_j)$ for all $i \le \ell \le j$ and replace $(a_1, \dots, a_k)$ by $(a_1, \dots, a_i, a_j, \dots, a_k)$.
	\end{itemize}

	\ls{We will prove in the next section}{In \cite{arxiv} we show} that $\tau(A)$ is a unique subsequence of $A$ (Property P1) and there are at most $2 \cdot 2^{2\given\tw}$ different typical sequences (Property P2). Hence, there are only a bounded number of typical sequences if their integers are bounded by $\given\tw$. We will call two integer sequences equivalent, if their typical sequences are the same. An extension of an integer sequence $(a_1, \dots, a_n)$ is a sequence of the form $(a_1, \dots, a_i, a_i, \dots, a_n)$ for $1 \le i \le n$, i.e.,~one integer is repeated. The set of all extensions of an integer sequence $A$ is denoted as $\E(A)$. Notice that an extension of a tree decomposition corresponds to an extension of the integer sequence assigned to one of the nodes of the compact core. Notice that all extensions $A^{\E} \in \E(A)$ of $A$ have the same typical sequence.

	Let $\td, \td' \in TD_{\given t}$ with $\td' \in \E^{\given t}(\td)$ and assume that we can use $\td$ in the construction of a tree decomposition of $G$ of width $\tw$. Notice that we can also use $\td'$ to construct such a tree decomposition $\td''$ (if we do not restrict to tree decompositions of at most $n$ nodes) since we can extend all tree decompositions used in the construction of $\td''$ at the same nodes as we extended $\td$ to obtain $\td'$. We call an integer sequence $A$ superior to $B$ if there are $A' \in \E(A)$ and $B' \in \E(B)$ such that $A' \le B'$. \ls{We will}{In \cite{arxiv} we} show that $A$ is superior to $\tau(A)$ (Property (P3)) and that $\tau(A)$ is superior to $A$ (Property (P4)). Hence, if we constructed a tree decomposition for which a node of the compact representation has the integer sequence $A$, we can assume that we have a tree decomposition with integer sequence $\tau(A)$ (even if no such tree decomposition exists) since for all tree decompositions that can be constructed from the latter, we can construct a tree decomposition of the same width as the first.

	\ls{}{Finally, we show in \cite{arxiv} that we can compute all typical sequences that can arise when adding up two extensions of two typical sequences (Property (P5)).}

	The compressed core of a tree decomposition is the compact representation of its core, where each node is assigned its typical sequence together with the size of the largest bag. We call two tree decompositions sequence-equivalent, denoted as $\equiv^{\given t}_{s}$, if their compressed cores are the same.
	
	In order to compensate the replacement of an integer sequence by its typical sequence, we allow arbitrary extensions of the compressed core. 
	\ls{Formally, we define $J^{\E}_{\given t}(\td^1, \td^2)=J(\E^{\given t}(\L^{\given t}(\td^1)), \E^{\given t}(\L^{\given t}(\td^2)))$, $I^{\E}_{\given t}(\td)=I_{\given t}(\E^{\given t}(\L^{\given t}(\td)))$ and $F^{\E}_{\given t}(\td)=F_{\given t}(\E^{\given t}(\L^{\given t}(\td)))$.}{}
	
	\ls{
	We now show how the construction in join-, introduce, and forget-nodes changes and that it suffices to enumerate one tree decomposition per equivalence class. In order to do so, we will use further properties on typical sequences, which we will show in the next Section.}{}
	
	\ls{
	\begin{construction}\label{ccconstruction} (Construction for all tree decompositions in $\E^{\given t}(\L^{\given t}(\td))$)
		\begin{description}
			\item[Join-node:] Notice that the extensions are operations that only affect the integer sequences assigned to the nodes of the compact representation of the core. Hence, we can consider these integer sequences independently. Consider a node $t$ of the compact representation of the core and let $A$ and $B$ be the integer sequences assigned to it in $\given t^1$ and $\given t^2$.
			In order to compensate extensions in a join, we have to be able compute all typical sequences in $\E(A)+\E(B)-|X_t \cap Y_{\given t}|$, where $\E(A)+\E(B):=\{A'+B' \mid A' \in \E(A), B' \in \E(B) \mbox{ and $A'$ and $B'$ have the same length} \}$, which we call property (P5). Notice that we can construct the compressed cores of the elements of $J^{\E}_{\given t}(\td^1,\td^2)$ from the compressed cores of $\td^1$ and $\td^2$.
			
			\item[Introduce-node:] As in the last construction, we have two cases: either $v$ is inserted to a bag whose node is in the core, in which case the core does not change, or $v$ is added to a bag whose node is not the core and the corresponding consistent path is now in the core.
			
			In the second case, all nodes on the added path with the same subset of $X_t$ are represented by a single node in the compact representation of the core and the integer sequences of these nodes are of the form $(a, \dots, a)$, where $a$ is the size of the subset. We simply replace this integer sequence by this typical sequence $(a)$.
			
			Now consider the case that $v$ is added to at least one node of the core. Let $t$ be a node of the compressed core that represents one of the nodes where $v$ is added, $(a_1, \dots, a_f)$ its integer sequence and $(a_{i_1}, \dots, a_{i_{f'}})$ its typical sequence. Let $t_1, \dots, t_f$ be the nodes represented by $t$. Notice that by Lemma \ref{compact}, we can assume that $v$ is inserted up a node $t_\ell$ with size $a_\ell$ such that there are no $i, j$ with $i < \ell < j$ and $a_i \le a_k$ and $a_j \ge a_k$ for all $i \le k \le j$. The remaining indices are part of the typical sequence, i.e.~$\ell=i_{\ell'}$ for some $\ell'$. Furthermore, we may assume that $a_{\ell+1}=a_\ell$ before the addition of $v$, as we can add a redundant node with bag $X_{t_{\ell}} \setminus \{v\}$ as a neighbor of $t_\ell$ to the final tree decomposition, and replace the edge $t_\ell t_{\ell+1}$ by the edge between $t_{\ell+1}$ and this new node. Hence, if $v$ is added to an initial subpath, say $(t_1, \dots, t_\ell)$, the integer sequences of the two nodes of the compact representation of the core are $(a_1+1, \dots, a_\ell+1)$ and $(a_\ell=a_{\ell+1}, \dots, a_k)$. The corresponding typical sequences are $\tau(a_{i_1}+1, \dots, a_{i_{\ell'}}+1)$ and $\tau(a_{i_{\ell'}}, \dots, a_{i_{f'}})$. We can argue similarly, if $v$ is added to a general subpath. 
			
			Notice that, again, the compressed cores of the constructed tree decompositions can be computed directly from the compressed core of the tree decomposition used in the construction.
			
			\item[Forget-node:] Notice that the removal of a vertex $v$ can result in the joining nodes from the compact representation of the core. Hence, we will have to be able to join typical sequences, i.e.,~$\tau(a_1, \dots, a_f, b_1, \dots, b_g)$ has to be computable from $\tau(a_1, \dots, a_f)$ and $\tau(b_1, \dots, b_g)$. When proving (P1), we will show that the order in which we apply the operations does not matter and thus $\tau(a_1, \dots, a_f, b_1, \dots, b_g)=\tau(\tau(a_1, \dots, a_f), \tau(b_1, \dots, b_g))$. Hence, the compressed cores can be computed directly from the compressed core of the tree decomposition used in the construction.
		\end{description}
	\end{construction}
}{We will not show the construction of $J^{\given t}(\E^{\given t}(\L^{\given t}(A)),\E^{\given t}(\L^{\given t}(B)))$, $I^{\given t}(\E^{\given t}(\L^{\given t}(A)))$ and $F^{\given t}(\E^{\given t}(\L^{\given t}(A)))$ but give the main points. First, we note that we can consider the integer sequences assigned to the nodes of the compact core independently. For a join-node, we have to compute all typical sequences of integer sequences obtained by adding two extensions of the given typical sequences which can by done by (P5). For an introduce-node, use the Lemma \ref{compact} by which we can assume that the vertex is either added to all nodes of the integer sequence or it is added up to a node which is contained in the typical sequence and this node can be assumed to be extended before. Hence, we can perform the necessary operations directly on the typical sequences. In a forget-node, we potentially have to join typical sequences if restricted bags become equal.}

	Notice that the compressed cores of the elements in $J^{\given t}(\E^{\given t}(\L^{\given t}(A)),\E^{\given t}(\L^{\given t}(B)))$, $I^{\given t}(\E^{\given t}(\L^{\given t}(A)))$ and $F^{\given t}(\E^{\given t}(\L^{\given t}(A)))$ can be directly computed from the compressed cores of $A$ (and $B$) and hence Lemma \ref{cancompress} holds for $\equiv^{\given t}_s$.

	\ls{	
	\subsection{The NeccesaryProperties of Typical Sequences}\label{typcial}

	In this section, we will show properties (P1) to (P5) of typical sequences, as well as properties (H1), (H2) and (H3), which are used in the proof. We want to mention that typical sequences are also of interest in different contexts (see \cite{DBLP:journals/corr/abs-1905-03643}).

	\begin{figure}
		\begin{center}
		\begin{tabular}{ccccccc}
			a) & 2 & 5 & 3 & 6 & 4 & 3 \\
			b) & 2 &   &   & 6 &   & 3 \\
			c) & 2 & 2 & 2 & 6 & 3 & 3 \\
			d) & 2 & 6 & 6 & 6 & 6 & 3
		\end{tabular}
		\end{center}
		\caption{An integer sequence (a), its typical sequence (b), the extension of the typical sequence that is smaller than the original sequence (c), and that is larger than the original sequence (d).}\label{figureextension}		
	\end{figure}
	
	\begin{description}
		\item[(P1)] The typical sequence $\tau(A)$ of $A$ is unique.\\
		Proof: Let $(a_1, \dots, a_k)$ be an integer sequence and assume $(a_1, \dots, a_i, a_j, \dots, a_k)$ obtained from $(a_1, \dots, a_k)$ by one operation. For $\ell \notin \{i+1, \dots, j-1\}$ we claim that $a_\ell$ can be removed by an operation from $(a_1, \dots, a_k)$ if and only if it can be removed by one operation from $(a_1, \dots, a_i, a_j, \dots, a_k)$. Using this iteratively, we see that no matter what the order of the operations is, it does not change whether we can remove an integer or not. In order to prove this claim, we have to slightly modify the second operation. If $\min(a_i,a_j) \le a_\ell \le \max(a_i,a_j)$ for all $i \le \ell \le j$, let $j^*$ be the smallest index such that $a_{\tilde j}=a_j$ for all $\tilde j \ge j^*$. We replace $(a_1, \dots, a_k)$ by $(a_1, \dots, a_i, a_{j^*}, a_{j+1}, \dots, a_k)$, i.e.,~instead of keeping $j$, we keep the same integer but another appearance (the first position such that afterwards all integers are the same as $a_j$). The modification makes the indices of the integers in the final typical sequence unique (and not just the values).
		
		The claim is obvious if $(a_1, \dots, a_i, a_j, \dots, a_k)$ is obtained from $(a_1, \dots, a_k)$ by an operation of the first type. Hence, consider the operation $o$ of the modified version of second type replacing $(a_1, \dots, a_k)$ by $(a_1, \dots, a_i, a_{j^*}, \dots a_j)$.
		
		Assume we can remove $a_\ell$ after $o$ is applied as $a_\ell=a_{\ell-1}$. Since $j^*$ cannot be $\ell$ in the operation, the claim follows, since we can still apply the same operation. Similarly, if $a_\ell$ can be removed before $o$ is applied as $a_\ell=a_{\ell-1}$, we can remove $a_\ell$ after $o$ is applied except if $o$ removes $a_\ell$.
		
		We will now show the claim for the second type of operation.
		
		Assume we can remove $a_\ell$ as there are $i' < \ell < j'$ such that $\min(a_{i'},a_{j'}) \le a_\ell \le \max(a_{i'},a_{j'})$ after $o$ is applied. As $i'$ and $j'$ are indices of $A$, we can also apply the operation before $o$ is applied.
		
		Assume we can remove $a_\ell$ as there is $i' < \ell < j'$ such that $\min(a_{i'},a_{j'}) \le a_\ell \le \max(a_{i'},a_{j'})$ before $o$ is applied. If $i'$ and $j'$ are in $(a_1, \dots, a_i, a_j, \dots, a_k)$, we can still apply the same rule after $o$ is applied. If $i'$ is removed, one of $a_i, a_j$ is smaller and the other larger as $a_{i'}$. Replacing $i'$ suitably by $i$ or $j$, we can again apply the rule after $o$ is applied. The same is true if $a_{j'}$ was removed.
		 
		\item[(H1)] If $(a_1, \dots, a_k)$ is a typical sequence, then $(a_k, \dots, a_1)$ is a typical sequence.\\
		Proof: Directly from the definition.
		 
		\item[(H2)] For an integer sequence $A$ with $\ell$ different values, the length of $\tau(A)$ is at most $2\ell$.\\
		Proof: We first show by induction that the length is at most $\ell$ if the sequence starts with the largest or smallest value and this value appears only once. If there is only one value, the typical sequence can only have one entry as we could otherwise apply the first operation. Assume now that we have at least two different values. We will discuss the case where the sequence starts with the largest value.
		In this case, the smallest value has to be unique and the second entry of the integer sequence (as we can remove the part between the largest value and the last appearance of smallest value by one application of the second operation). Hence, the sequence that starts with the second entry is an integer sequence that starts with the smallest value, that has at most $\ell-1$ different values, and where the smallest value appears only once. Using the induction hypothesis, we see that its length is at most $\ell-1$. Hence, including the first entry, we have a sequence of length at most $\ell$.
		
		Now consider an arbitrary typical sequence $(a_1, \dots, a_k)$. If all integers are the same, the typical sequence can have only one entry and, hence, a length of $1$. Assume now that we have at least two different values. Let $i$ and $j$ be indices of a smallest and largest entry and assume w.l.o.g.~$i < j$. It holds that $j-i=1$ and the smallest and largest integer appear only once, as we can remove all integers between the first appearance of $a_i$ and the last appearance of $a_j$ by one application of the second operation. Hence $(a_1, \dots, a_k)$ 
		is composed of an integer sequence that ends with the unique smallest value followed by an integer sequence starting with the unique largest value. Using the claim above and (H1), we find that its length is at most $\ell+\ell=2\ell$.
		
		\item[(P2)] There are at most $2 \cdot 2^{2\ell}$ typical sequences with $\ell$ different values.\\
		By the proof of (H2) we know that the largest and the smallest value are neighbored. We show that there are at most $2^{2\ell}$ typical sequences where the smallest value is immediately before the largest value. As the same number can be proven if the largest value is immediately before the smallest, the claim follows.
		
		We count the number of typical sequences up to the smallest value and starting from the largest value, starting with the second term.
		Again by the proof of (H2) we conclude that in a typical sequence starting with the largest value each value appears at most once. Furthermore, the typical sequence is uniquely determined by the set of values that appear in the sequence, as the following entries have to be alternating the smallest and largest of the remaining values. Hence there are at most $2^{\ell}$ such sequences. The same holds for the number of typical sequences ending with the smallest value and hence there are at most $(2^\ell)^2=2^{2\ell}$ typical sequences of where the largest value immediately follows the smallest value.
		
		\item[(P4)] The typical sequence $\tau(A)$ is superior to $A$.\\
		Proof: $\tau(A)$ is obtained from $A$ by iteratively removing parts of $A$. We show the claim by induction on the number of removals. Assume $(a_1, \dots, a_k)$ is replaced by $(a_1, \dots, a_i, a_j, \dots, a_j)$. We extend the smaller of $a_i$ and $a_j$ for $j-i-1$ times to obtain a sequence of length $k$ that is at most $(a_1, \dots, a_k)$ (see Figure \ref{figureextension}).
		\item[(P3)] $A$ is superior to its typical sequence $\tau(A)$.\\
		Proof: Similarly to the above, but we extend the larger value of $a_i$ and $a_j$ to obtain a sequence that is at least $(a_1, \dots, a_k)$ (see Figure \ref{figureextension}).
		\item[(H3)] The typical sequence of $(a_1-c, \dots, a_k-c)$ is obtained from $\tau(a_1,\dots, a_k)$ by subtracting $c$ from each number.\\
		Proof: This property follows directly from the definition of typical sequences.
		\item[(P5)] We have to show that we can compute $\tau(\E(A)+\E(B)-c)$. Using (H3), it suffices to show how to compute $\tau(\E(A)+\E(B))$. The construction is as follows:
		 
		For typical sequences $A=(a_1, \dots, a_k)$ and $B=(b_1, \dots, b_\ell)$, let $G^{A,B}=(V^{A,B},E^{A,B})$ be the directed graph with nodes $v_{i,j}$ for $1 \le i \le k$ and $1 \le j \le \ell$ and edges $E^{A,B}=\{(v_{i,j},v_{i',j'}) \mid 1 \le i \le i' \le k, 1 \le j \le  j' \le \ell, i'-i \le 1, j'-j\le 1, (i,j) \not= (i',j')\}$ and let $w(v_{i,j})=a_i+b_j$. Let $P(A,B)= \{ (w(u_1), \dots, w(u_{m})) \mid (u_1, \dots, u_{m}) \mbox{ be the nodes on a}$ $\mbox{path in $G$ from $v_{1,1}$ to $v_{k, \ell}$ } \}$ (see Figure \ref{typicalpath}). 
		
		$P(A,B)$ is computable since we can enumerate all paths of a graph. We claim that $\tau(P(A,B))=\tau(\E(A)+\E(B))$.

		\begin{figure}
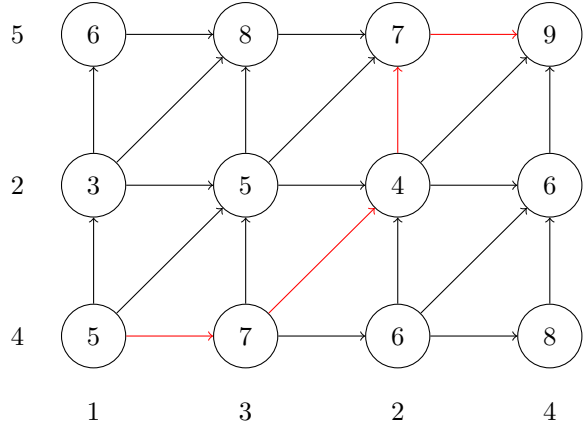

			\begin{center}
			\gab
			\end{center}
			\caption{The graph $G^{AB}$ constructed for the sequences $A=(1,3,2,4)$ and $B=(4,2,5)$. The path in red corresponds to the extensions $(1,3,2,2,4) \in \E(A)$ and $(4,4,2,5,5) \in \E(B)$.}\label{typicalpath}
		\end{figure}		
		
		Proof: We first show that each sequence of $P(A,B)$ is an element of $\E(A) + \E(B)$. Consider a path $v_{i_1,j_1}, \dots, v_{i_m,j_m}$ between $v_{1,1}$ and $v_{k,\ell}$ in $G^{A,B}$. Let $A'=(a_{i_1}, \dots, a_{i_m})$ and $B'=(b_{j_1}, \dots, b_{j_m})$. Notice that $A' \in \E(A)$ as $i_1, \dots, i_m$ are non decreasing and increase in each step at most by one. Similarly $B' \in \E(A)$. By definition $A'+B'=(w(v_{i_1,j_1}), \dots, w(v_{i_m, j_m}))$.
		
		Now consider arbitrary extensions $A'=(a_{i_1}, \dots, a_{i_m}) \in \E(A)$ and $B'=(b_{j_1}, b_{j_m}) \in \E(B)$ of the same length. If $i_k=i_{k+1}$ and $j_k=j_{k+}$, we have $(A'+B')_k=(A'+B')_{k+1}$ and we can delete index $k+1$ from $A'$ and $B'$ without changing the typical sequence of $A'+B'$. Hence, we can assume that $i_k \le i_{k+1} \le i_k+1$, $j_k \le j_{k+1} \le j_k+1$ and $(i_k,j_k) \not= (i_{k+1}, j_{k+1})$. Furthermore, $i_1=j_1=1$, $i_m=k$ and $j_m=\ell$. Hence $(v_{i_1,j_1}, \dots, v_{i_m,j_m})$ is a path in $G^{A,B}$ and $(w(v_{i_1,j_1}), \dots, w(v_{i_m,j_m}))=A'+B'$.
	\end{description}
	
}{For proof of the necessary properties of typical sequences, we refer to the extended version of this paper\cite{arxiv}. Furthermore, we estimate the number of compact cores enumerated for each bag and the running time in that paper. Here, we only remark that it is not hard to see that this number and the running time can be bounded by a function computable from $\given\tw$. }

\ls{	
	\subsection{Number of equivalence classes and computation time}
	
	The number of equivalence classes we enumerate in a bag $Y_{\given t}$ of the given tree decomposition is bounded by the number of different compressed cores for the bag. Recall that a compressed core is a tree with at most $2\given\tw^2$ nodes where a restricted bag assigned to each node and a typical sequence is assigned to each edge. 
	The number of equivalence classes can be bounded by the product of the number of tree topologies, the number of ways to assign subset of vertices to the nodes of the tree and the number of ways to assign typical sequences to the edges of the tree.
	
	Although much better estimates for the number of tree topologies are known, we use the following crude estimate: a tree of at most $2\given\tw^2$ nodes can be described by a graph over exactly $2\given\tw^2$ nodes by given the edges of the tree and ignoring the nodes without adjacent edges. The number of graphs over $2\given\tw$ nodes is $2^{(2\given\tw)^2}=2^{4\tw^2}$
	
	A restricted bag is characterized by the subset of the vertices of $Y_{\given t}$ it contains and its size and hence the number of choices for a restricted bags is at most $2^{\given \tw} \cdot \given\tw \le 2^{2\given\tw}$. Choosing a restricted bag for each node gives $(2^{2\given\tw})^{2\given\tw^2}=2^{4\given\tw^3}$ possibilities.
	
	In (P2) we proved that there are at most $2 \cdot 2^{2\given\tw} = 2^{2\given\tw+1} \le 2^{3\given\tw}$ typical sequences with values bounded by $\given \tw$. Choosing a typical sequence for each edge gives at most $(2^{3\given\tw})^{2\given\tw^2} \le 2^{6\given\tw^3}$ possibilities.
	
	Hence, in total we can bound the number of compressed cores by $2^{4\given\tw^2} \cdot 2^{4\given\tw^3} \cdot 2^{6\given\tw^3} = 2^{4\given\tw^2+4\given\tw^3+6\given\tw^3} \in 2^{\Oh(\tw^3)}$. 
	
	We now want to argue that the running time to construct the compressed cores is polynomial in the bound on their number which then leads to a total running time of $(2^{\Oh(\given\tw^3)})^c \cdot n=2^{\Oh(\given\tw^3)} \cdot n$ where $c$ is the exponent of the polynomial.	Notice that although many of these estimates seem to be quite crude, no exact algorithm with a running time bound of the form $2^{o(\tw^3)} \cdot n^{\Oh(1)}$ is known.
	
	For a join-node $\given t$, we construct all compressed cores that can build for a pair $C_1,C_2$ of compressed cores of the two children and add it to the enumerated compressed cores if it is not already contained (which can be tested by simply iterating over all compressed cores that are already enumerated). In order to do so, we have to check that the tree topologies and the vertices of the restricted bags match in $C_1$ and $C_2$. We compute the sizes of the restricted bags in the resulting compressed core. Furthermore, for each edge of the tree, we compute all possible typical sequences that can arise by joining the corresponding typical sequences of $C_1$ and $C_2$ by enumerating all path in the graph constructed to prove (P5) and reduce the corresponding integer sequences to typical sequences. Notice that although the number of path is exponential in the size of the graph, it is polynomial in the number of possible compressed cores.
	
	For a forget-node $\given t$, we iterate over all compressed cores enumerated in the child of $\given t$, remove the forgotten vertex of all restricted bags and computing resulting compressed core. If this is not already enumerated, we add it to the enumerated compressed cores.
	
	For an introduce-node $\given t$, the simplest way to describe a construction is to invert the approach of a forget-node. We enumerated all possible compressed cores $C$ for $\given t$. For each, we compute the compressed core $C'$ that results in removing the introduced vertex from all bags in $C$ and test whether we enumerated $C'$ in the child of $\given t$. A constructive approach would select a compressed core $C'$ in the child and then either add the introduced vertex into a subtree of the tree of $C$ or to a single new bag and join this bag by a path of nodes whose bags are strictly nested subsets to one of the nodes of the tree of $C'$. Each integer sequence of a new edge consists of only one entry whose value is the size of the bag.

}{}
\ls{
	\section{The Algorithm to Compute the Optimal Tree Decompositon}
	
	In this section, we will show how we can use the algorithm described above to compute an optimal tree decomposition without a given tree decomposition. We start with a simple variant, leading to an algorithm with quadratic running time (in the number of vertices of $G$ for a graph with an arbitrary but fixed treewidth) and then give a more complicated algorithm with linear running time.
	
	In this section, we give an algorithm to test whether the treewidth is at most $\tw$ for a given value $\tw$. We can compute the treewidth from such an algorithm by testing increasing values for $\tw$ starting with $1$ and increasing the value until we find the treewidth. Hence, the algorithm outlined is called exactly $\tw$ times until we find the correct value. Let $n=|V|$.
	
	\subsection{Simple Version}
	
	We use the fact that the treewidth of a minor of a graph is at most that of a graph itself (see Section \ref{basic}).
	
	We assume that the given graph $G$ is connected since we can otherwise compute tree decompositions for each connected component independently, as stated in Section \ref{basic}.
	
	Let again $G=(V,E)$ be the given graph. 
	If $G$ contains only one vertex, we are done. Otherwise there is at least one edge in $G$. Take an arbitrary edge $uv$ and shrink it and let $[uv]$ be the new node and $G^{uv}$ be the resulting graph (which is also connected). Notice that we compute a minor and hence the treewidth of the $G^{uv}$ is at most that of $G$. Recursively compute a tree decomposition of $G^{uv}$ of width $\tw$. From the computed tree decomposition, construct a tree decomposition of $G$ of width at most $\tw+1$ by replacing each appearance of the vertex $[uv]$ by the two vertices $u$ and $v$. Using the algorithm of Section \ref{KloksAlg}, we compute a tree decomposition of $G$ of width $\tw$.
	
	For every fixed value of $\tw$, the running time $T(n)$ is $T(n) \le T(n-1) + \Oh(n)=\Oh(n^2)$ as the running time of the algorithm in Section \ref{KloksAlg} is $\Oh(n)$ for every fixed value of $\tw$.
	
	\subsection{Improved Version}

	The following algorithm is based mainly on ideas from Perkovic and Reed\cite{DBLP:journals/ijfcs/PerkovicR00} and only slightly simplified.
	
	We use the fact that a pair of vertices $u,w$ with more than $\tw$ common neighbors are contained in at least one common bag and, hence, the graph where we add the edge $(u,w)$ has the same treewidth (see Section \ref{basic}).
	
	We first test if the graph has treewidth $1$ (i.e.,~is a graph without edges) or $2$ (i.e.,~the graph is a tree) in linear time. In the following, we will assume that the treewidth is at least $3$.
	
	Assume we want to compute a tree decomposition of width $\tw$ of a graph $G=(V,E)$ or show that the treewidth is at least $\tw+1$. We will show how we can find a graph $G'$ of at most $(1-1/16\tw^2)|V|$ nodes of treewidth at most $\tw$, such that we can construct a tree decomposition of $G$ of width at most $2\tw$ from a tree decomposition of $G'$ of width $\tw$. We compute a tree decomposition of $G'$ of width $\tw$ recursively and, from it, construct a tree decomposition of $G$ of width at most $2\tw$. Finally, we use the algorithm of Section \ref{KloksAlg} for a tree decomposition of $G$ of width $\tw$.
	
	The total running time for a graph with $n$ nodes can be computed as 
	$T(n) \le T((1-1/16\tw^2)n)+\Oh(n)=\Oh(n)$ for every fixed value of $\tw$. 
	
	We will now show how to find $G'$. In the algorithms, we assume an arbitrary but fixed numbering of the nodes. We first recursively remove all nodes of degree one (since a tree decomposition of the remaining graph can easily be extended to a tree decomposition of the original graph of the same width).
	
	Let $S$ be the nodes of degree at most $4\tw$. We compute a maximal matching $M$ in $G$ restricted to the edges, such that at least one endpoint is in $S$. If $|M| \ge |V|/16\tw^2$, we shrink the edges obtaining a minor $G'$ of $G$ with at most $(1-1/16\tw^2)|V|$ nodes and a treewidth of at most $\tw$. If $\td$ is a tree decomposition of $G'$, we can easily obtain a tree decomposition of $G$ by replacing each occurrence of a node $[uv]$ created when shrinking a matching edge $uv$ by the two nodes $u$ and $v$. The width of this tree decomposition is at most twice the width of $\td$.
	
	\begin{figure}
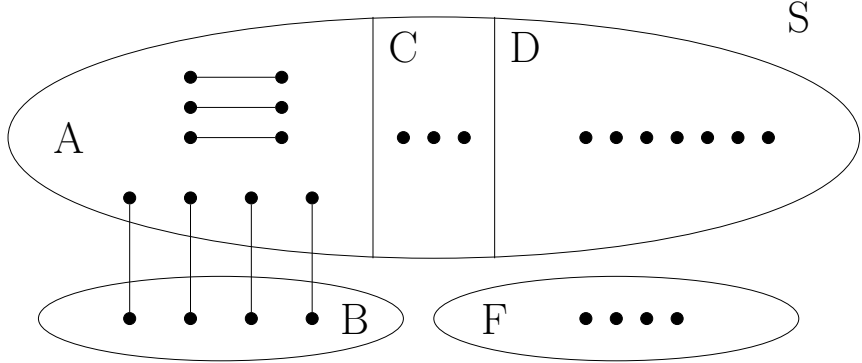

		\begin{center}
			\scalebox{0.8}{\figeins}
		\end{center}
		\caption{The nodes $S$ of degree at most $4\tw$ are partitioned into the nodes $A$ that are matched, their neighours $C$, and the remaining nodes $D$. The nodes from $V \setminus S$ are partitioned into the nodes $B$ that are matched, and the remaining nodes $F$.}\label{figureperkovic}
	\end{figure}
		
	Assume now that $|M| < |V|/16\tw^2$. 
	Let $A$ be the matched nodes of $S$ and $B$ be the nodes in $V \setminus S$ that are matched to a node in $A$. Let $C$ be the neighbors of $A$ within $S$, $D=S \setminus A \setminus C$ and $F=V \setminus S \setminus B$ (see Figure \ref{figureperkovic}). Notice the following: 
	\begin{itemize}
		\item $|S| \ge |V|/2$: As the treewidth of $G$ is at most $\tw$, the number of edges of $G$ is at most $|V|\tw$ and, hence, the number of nodes of degree more than $4\tw$, i.e.,~the number of nodes not in $S$, is at most $|V|/2$.
		\item $|A| \le |V|/8\tw^2$: since $|M| < |V|/16\tw^2$ and each edge in $M$ can enforce at most two nodes in $A$
		\item $|B| \le |V|/16\tw^2$: since each edge in $M$ can enforce at most one node in $B$
		\item $|C| \le |V|/2\tw$: since each node in $A$ has at most $4\tw$ neighbors (in $C$).
		\item $|D| = |S|-|A|-|C| \ge |V|/2-|V|/8\tw^2-|V|/2\tw \ge |V|/5$, where the last estimate uses the fact that $\tw \ge 2$.
		\item All neighbors of a node $u \in D$ are in $B$ since all neighbors are matched (otherwise we could increase the matching), and there are no neighbors in $A$ by construction. As we assume that all nodes have degree at least $2$, there are at least $2$ neighbors. 
	\end{itemize}
	
	For $v \in V$, let $N(v)$ be the neighbors of $v$ in $G$.
	Now we try to assign each node $w \in D$ to a pair $(u,v)$ of distinct nodes in $N(w)$, such that no pair of nodes has more that $\tw+1$ nodes from $D$ assigned to it. Notice that $N(w) \subseteq B$  for all $w \in D$ and hence we are considering pairs of nodes in $B$ here. We do this greedily until no further assignment is possible, i.e.,~if $U$ is the set of nodes that is not assigned to a pair, all pairs of neighbors of a node in $U$ already have $\tw+1$ nodes assigned to them. We do this in linear time as follows: We make a list of pairs of nodes of $B$ in which we insert for each $w \in D$ all pairs of neighbors of $w$. As $w$ has at most $4\tw$ neighbors, this list will contain $\Oh(|D|\tw^2) \subseteq \Oh(|V|)$ pairs of nodes. We sort it lexicographically. We iterate over the list and assign vertices to pairs of the list if they are not assigned yet and the pair has assigned less than $\tw+1$ vertices so far. 
	
	As stated in Section \ref{basic}, we know that for each $w \in U$, we can add edges $(u,v)$ to $G$ for all pairs of neighbors of $w$ without increasing the treewidth of the graph. Let $G^*$ be the graph obtained by adding these edges for all nodes in $U$ and removing the nodes of $U$ and all adjacent edges. The neighbors with respect to $G$ of $w$ form a clique in $G^*$ and hence there will be a bag containing these nodes in each tree decomposition of $G^*$. We can construct a tree decomposition of $G$ from a tree decomposition of $G^*$ by adding a bag for each $w \in U$ that contains $\{w \} \cup N(w)$ and attaching it to an arbitrary bag containing $N(w)$. We will later discuss how we can efficiently find these bags for the nodes in $U$. We will now show that $|U| \ge |V|/8\tw^2$ and, hence, we can recursively search for a tree decomposition of $G^*$ (as in the case where we have a large matching).
	
	Assume $U$ contains less than $|V|/8\tw^2$ nodes. Consider the graph $\hat G$ obtained as follows: Start with $G[B \cup (D \setminus U)]$ and for each $w \in D \setminus U$ shrink $w$ to either $u$ or $v$ (chosen arbitrarily) if $(u,v)$ is the pair of nodes assigned to $w$. Notice that the nodes of $\hat G$ are $B$ and $uv$ will be an edge of $\hat G$ if the pair $(u,v)$ was chosen for a node $w \in D \setminus U$. As $|D \setminus U| \ge |V|/8$ and for each pair of nodes $(u,v)$ at most $\tw+1$ nodes in $D \setminus U$ are assigned to it, $\hat G$ has at least $|V|/8(\tw + 1)$ edges. As the number of edges of a graph of treewidth $\tw$ with $n$ vertices is at most $\tw n$, as stated in Section \ref{basic}, the treewidth of $\hat G$ is at least $\tw+1$ since $\hat{G}$ has at most $|V|/16\tw^2$ nodes and at least $|V|/8(\tw+1)>|V|/16\tw$ edges. This is a contradiction since $\hat G$ is a minor of $G$.
	
	We will now discuss how we find a bag containing $N(w)$ for all $w \in U$ in linear time. We make a list that contains all subsets of vertices of every bag with a link to the originating bag. The size of the list is at most the number of bags times $2^{\tw+1}$, which is in $\Oh(n)$ and each entry is a subset of at most $\tw$ nodes and hence has size $\Oh(1)$. We sort this list lexicographically. Furthermore, we lexicographically sort the list containing $N(w)$ for all $w \in U$. We can now find a bag that contains $N(w)$ by iterating over the two lists in parallel (as in the merge of merge-sort).	
	
}{}

	\ls{}{\vspace*{-1.8em}}
	\section{Conclusion}
	
	We gave simpler descriptions of two algorithms: The algorithm of Bodlaender and Kloks \cite{DBLP:journals/jal/BodlaenderK96} that computes an optimal tree decomposition for a graph $G$ given a (non-optimal) tree decomposition of bounded width in linear time, and the algorithm of Bodlaender \cite{DBLP:journals/siamcomp/Bodlaender96} that uses the first algorithm to compute an optimal tree decomposition of a graph with bounded treewidth in linear time.
	
	Although we were able to shorten the text significantly, the description is still too long to become part of textbooks. We hope that even simpler descriptions of the algorithms can be found in the future, which will finally allow these algorithms to be shown in textbooks.
	
 \ls{}{\vspace*{-1.8em}}
	  
  \ls{\bibliography{references}}{\bibliography{references_nodoi}}
	  	
\end{document}